\documentclass[12pt,journal,draftclsnofoot,onecolumn,english]{IEEEtran}
\usepackage{cite}
\usepackage{amsmath,amssymb,amsfonts}
\usepackage{epsfig,subfigure,multirow}
\usepackage{algorithmic}
\usepackage{graphicx}
\usepackage{textcomp}
\usepackage{xcolor, color}
\usepackage{array}
\usepackage{pstricks}
\usepackage{verbatim}
\usepackage{stfloats}

\usepackage[bookmarks=false]{}
 \usepackage{aecompl}
\usepackage{mathrsfs}
\usepackage{arydshln}
\usepackage{mathtools} 
\usepackage{amsthm}
\usepackage{caption}
\usepackage{lipsum}
\usepackage{comment}
\usepackage{tabularx}
\usepackage{changepage}
\usepackage{tikz} 
\usetikzlibrary{shapes.geometric}
\usetikzlibrary{positioning}

\usepackage[utf8]{inputenc}

\def\BibTeX{{\rm B\kern-.05em{\sc i\kern-.025em b}\kern-.08em
    T\kern-.1667em\lower.7ex\hbox{E}\kern-.125emX}}

\newtheorem{theorem}{Theorem}
\newtheorem{corollary}{Corollary}
\newtheorem{definition}{Definition}
\newtheorem{lemma}{Lemma}

\newtheorem{remark}{Remark}

\newtheorem{example}{Example}

\allowdisplaybreaks
\begin{document}

\title{Coded Caching with Private Demands and Caches}

\author{
Ali~Gholami,~\IEEEmembership{Student Member,~IEEE,} 
Kai~Wan,~\IEEEmembership{Member,~IEEE,} 
Hua~Sun,~\IEEEmembership{Member,~IEEE,}
Mingyue~Ji,~\IEEEmembership{Member,~IEEE,}  
and~Giuseppe Caire,~\IEEEmembership{Fellow,~IEEE}
\thanks{
The preliminary version   of this paper was presented in parts at 
    the 2022  IEEE  International Symposium on Information Theory, Espoo, Finland,~\cite{gholami2022coded}.}
\thanks{
A.~Gholami and G.~Caire are with the Electrical Engineering and Computer Science Department, Technische Universit\"at Berlin, 10587 Berlin, Germany (e-mail:  \{a.gholami, caire\}@tu-berlin.de). The work of A.~Gholami  and G.~Caire was partially funded by the European Research Council  under the ERC Advanced Grant N. 789190, CARENET.}
\thanks{
K. Wan was with   the Electrical Engineering and Computer Science Department, Technische Universit\"{a}t Berlin,
10587 Berlin, Germany. He is now with the School of Electronic Information and Communications, 
Huazhong University of Science and Technology, 430074  Wuhan, China,  (e-mail: kai\_wan@hust.edu.cn). The work of K.~Wan was partially funded by the   National Natural
Science Foundation of China (NSFC-12141107) and the CCF-Hikvision Open Fund (20210008).}
\thanks{
H.~Sun is with the Department of Electrical Engineering, University of North Texas, Denton, TX 76203 (email: hua.sun@unt.edu). The work of H.~Sun is supported in part by funding from NSF grants CCF-2007108 and CCF-2045656.
}
\thanks{
M.~Ji is with the Electrical and Computer Engineering Department, University of Utah, Salt Lake City, UT 84112, USA (e-mail: mingyue.ji@utah.edu). The work of M.~Ji was supported in part by NSF Awards 1817154, 1824558.}
}

\maketitle

\begin{abstract}
	This paper studies the privacy issue in coded caching. Recently it was shown that 	the seminal MAN coded caching scheme   leaks the demand information of each user to the other users in the system. 
	Many works have considered coded caching with demand privacy, while each non-trivial existing coded caching scheme with private demands was built on the fact that the cache information of each user is private to the others. 
	However, most of these schemes leak the users’ cache information. As a consequence, in most realistic settings (e.g., video streaming), where the system is used over time with multiple sequential transmission rounds, these schemes leak demand privacy beyond the first round.  
 This observation motivates  our new formulation of coded caching with simultaneously private demands and caches in this paper. 
For this new model, we first show that an existing coded caching scheme with private demands, referred to as the virtual users scheme, can also preserve the privacy of the users' caches.   However, this scheme suffers from its extremely high subpacketization. The main contribution of this paper is a new construction that generates private coded caching schemes  by leveraging two-server private information retrieval (PIR) schemes. 
  We show that if  in the PIR scheme the demand is uniform over all files and the queries  are independent, the resulting caching scheme is private on both the demands and on the caches; otherwise, the resulting scheme is private only on the demands. This first results construct the coded caching schemes from a particular class of  PIR schemes, which is a new ``structural’’ result in its own merit.
  We then construct new two-server PIR schemes with uniform demand and independent queries, such that the resulted caching scheme has a subpacketization level which is significantly reduced  compared to the virtual users scheme. 
 Interestingly   we propose  a new construction of two-server PIR schemes with uniform demand and independent queries by leveraging coded caching schemes. By applying the seminal Maddah-Ali and Niesen coded caching scheme into our construction, the resulting two-server PIR scheme   is proved to be order optimal under the constraint of  uniform demand and  independent queries.  
This is a second new ``structural’’ result, somehow closing the loop in the relation between coded caching and PIR.
As a by-product of our new construction, we obtain a new caching scheme with private demands that improves the load of the state-of-the-art demand private caching schemes known so far. 
  Finally, 
to explore a broader tradeoff between cache privacy and transmission load,
we relax the cache privacy constraint and introduce the definition of leakage on cache information. Then, again as a by-product of our new construction, we propose new schemes with perfect demand privacy and imperfect cache privacy that achieve an order-gain in load with respect to the scheme with perfect privacy on both demands and caches. This also establishes a first non-trivial achievability result in the   tradeoff between load and cache privacy, for demand-private caching schemes.   
  
    
	
\end{abstract}


\section{Introduction}
	Coded caching was first introduced in \cite{maddah2014fundamental}. In a caching system, the goal is to leverage the local memory available at the end-users to reduce the load in the network by exploiting the content already availbale in the cache rather than downloading from the server. Before the emergence of \cite{maddah2014fundamental}, this leverage was limited to the \textit{local caching gain} which depends on the local cache size. The coded caching scheme proposed by Maddah-Ali and Niesen, referred to as the MAN scheme in \cite{maddah2014fundamental}, showed that the cache memory available for every user can be used in an aggregate manner even if there is no cooperation between the users. This gain is referred to as the \textit{global caching gain}. Thus, in addition to benefiting from a local caching size, the system can benefit from the aggregate cache size which scales with the number of users and yields a larger further reduction of the network load.

	In the MAN coded caching setting \cite{maddah2014fundamental}, a server which has a library of $N$ files and is connected to $K$ users via a shared medium. Each user has a cache of size $M$ files. 
	A coded caching scheme consists of two phases: \textit{placement} and \textit{delivery}. In the placement phase, each user fills its cache   
	 without  knowledge of the users' later demands. If each user directly stores some bits of files into the cache, the placement phase is called {\it uncoded}. 
	 In the delivery phase, each user demands one file. According to the users' demands and caches, the server broadcasts multicast messages to the users such that each user can recover its demanded file. 
 The transmission load is defined as  the number of broadcasted bits in the delivery phase normalized by the file size. The objective is to minimize the worst-case load  among all possible demands.
	The  MAN coded caching scheme is based on a combinatorial design in the placement which splits each file into multiple subfiles and assigns each subfile to a subset of users, 
such that each multicast message is useful to $1+KM/N$ users. The achieved load by the MAN coded caching scheme is $\frac{K(1-M/N)}{1+KM/N}$, where 
	  $1-M/N$ represents the \textit{local} caching gain and $\frac{1}{1+KM/N}$ represents the \textit{global} caching gain. When $N\geq K$, the MAN scheme was proved to be order optimal within a factor of $2$~\cite{yu2017exact} and optimal under the constraint of uncoded cache placement \cite{wan2016optimality}.  When $N<K$, an improved coded caching scheme was proposed by Yu, Maddah-Ali, and Avestimhr (YMA) in~\cite{yu2017exact}, which is built on the fact that some MAN multicast messages can be re-constructed by the other ones and are thus redundant. The YMA scheme was then proved to be order optimal within a factor of $2$~\cite{yu2017exact} and optimal under the constraint of uncoded cache placement~\cite{yu2017exact}, for any system parameters. 
Following the seminal work of MAN, coded caching was studied in different extensions, including decentralized setting \cite{maddah2014decentralized}, online coded caching \cite{pedarsani2015online},   Device-to-Device (D2D) networks \cite{ji2015fundamental}, random and nonuniform   demands \cite{ji2017order,niesen2016coded},    hierarchical coded caching \cite{karamchandani2016hierarchical}, etc.  
	
	The MAN centralized scheme has a  subpacketization level at most exponential in the number of users $K$, which is one of its practical limitations.  Also for their decentralized algorithm in~\cite{maddah2014decentralized}, the multiplicative caching gain appears in the asymptotic regime of file size scaling to infinity. The authors in~\cite{shanmugam2016finite} addressed this issue and showed that this multiplicative gain is non-existent in the finite file size regime under random placement and clique cover delivery schemes. In order to reduce the subpacketization, the authors in~\cite{jin2019new} introduce a decentralized scheme that achieves a low worst-case load in the finite file size regime and maintain optimal memory-load tradeoff when file size scales to infinity. A combinatorial  structure, referred to as placement delivery array (PDA), was proposed in~\cite{yan2017placement} to design coded caching schemes with uncoded cache placement and clique-covering delivery, where the MAN scheme can be also seen as a coded caching scheme under PDA construction.  Following~\cite{yan2017placement}, various PDA constructions were proposed in~\cite{wang2019placement, yan2018placement, sasi2021multi, cheng2021framework, zhong2020placement}. Other combinatorial structures, such as hypergraphs~\cite{shangguan2018centralized},  Ruzsa-Szeméredi
graphs~\cite{shanmugam2017coded}, the strong edge coloring of bipartite graphs~\cite{yan2017placementbipartite},  linear block codes~\cite{tang2018coded},   have also been used to construct coded caching schemes with reduced subpacketization compared to the MAN scheme.
Under PDA construction, the subpacketization of the MAN scheme is minimum to achieve the load $\frac{K(1-M/N)}{1+KM/N}$~\cite{somevariantCheng}.

	\subsection{Demand private coded caching schemes}
	Despite the optimality guarantee, another issue of the MAN scheme is its leakage on the users' demand information. In order to decode the MAN multicast messages, each user should be aware of the other users' demands, which hurts the demand privacy. 	
Information theoretic formulation on coded caching with private demands was proposed in~\cite{wan2020coded}, where each user has a cache private to the other users and the privacy constraint requires that 
each user cannot obtain any information about other users' demands from the broadcasted messages in the delivery phase. 
Based on the virtual user strategy in~\cite{engelmann2017content}, an information theoretic private scheme was proposed in~\cite{wan2020coded}. By introducing $KN-K$ virtual users and letting each file be demanded by $K$ effective users (i.e., real or virtual users), 
the problem can be solved by using the $NK$-user MAN or YMA scheme.  
The resulting scheme can perfectly preserve the privacy of each user's demand against the other users, because each user cannot distinguish the real users from the effective users. 
The achieved load of this virtual user based scheme was proved to be order optimal within a constant except for the case $N\geq K$ and $M<N/K$. However, 	
	 its subpacketization is  at most exponential to $NK$ which is far from being practical.  In order to reduce the subpacketization, 
the authors in~\cite{wan2020coded} proposed another private caching scheme based on Minimum Distance Separable (MDS) codes for the case $M\geq N/2$, which achieves an order optimal load with subpacketization at most exponential to $K$. 
A varaint (but equivalent) virtual user based scheme was proposed in~\cite{DBLP:journals/corr/abs-1909-03324}, where for each real user we introduce $N-1$ virtual users such that the union set of the requested files by these $N$ effective users is the whole library.

	Following the coded caching problem with private demands, some improved schemes were proposed in order to reduce the load or subpacketization. 
	In \cite{aravind2020coded}, the authors proposed a demand-private scheme for the special case of a caching system with $N=2, K=2$ and $M=1$ while the  subpacketization level equal to $3$ was proved to be the minimum.  
A strategy introducing the use of  \textit{private keys} was proposed by Yan and Tuninetti in~\cite{yan2021fundamental}, whose main idea is to transform file retrieval to scalar linear function retrieval (i.e., each user requests a scalar linear function of files~\cite{wan2021optimal}). 
Each user's cache is  split into two parts. In the first part, every user caches the same subfiles as in the MAN scheme. The second part serves as the private key of each user,  composed of some linear combinations of the subfiles which are not cached in the first part. 
In the delivery phase, each user pretends to request a scalar linear function of the files, from which and the private key the user can recover its demanded file. 
Then by using the cache-aided scalar linear function retrieval scheme in~\cite{wan2021optimal}, the resulting private caching scheme requires a subpacketization which is the same as the MAN scheme.  The resulting scheme was proved to be order optimal within a factor of $ 6.3707$ if the metadata (i.e., the composition) of the broadcast messages is given.\footnote{\label{foot:equivalent}{In the privacy constraint of~\cite{yan2021fundamental}, the mutual information is under the condition of the realization of the library; this is equivalent to the case that the metadata  of the broadcast messages is provided in the header of the delivered packet, which is common in practice for the  ease of decoding.}}
 Other works on demand private caching can be found in \cite{kamath2020demand}, which proves that the optimal loads with and without demand privacy are within a multiplicative factor and also characterizes the exact memory-load tradeoff for the case $N=K=2$. In \cite{gurjarpadhye2022fundamental}, the authors provided the exact memory-load tradeoff for demand private coded caching when $N\geq K=2$. Finally in \cite{aravind2020subpacketization}, demand private coded caching was studied with the focus on reducing the subpacketization level. For the cases $N=K=2$, the authors proposed a scheme with lowest possible subpacketization. 
	
	\subsection{Brief review of private information retrieval (PIR)}
	The demand privacy was originally considered  in the PIR problem~\cite{chor1995private}, where 
	a user is connected to $S$ servers through $S$ individual private links, respectively. The library contains $N$ equal-length messages and  the user wants to   retrieve one message  from   the servers without letting the servers know any information about the demand. For this purpose, the user sends a query to each server and the server replies with an answer including some coded packets to the user. 	 
	The communication cost, defined as the amount of information exchanged between the user and the servers, is equal to the sum of the total upload cost (i.e., sum of individual upload costs defined as the length of the query from the user to the servers) and total download cost (i.e., sum of individual download costs defined as the length of the answer from the servers to the user normalized by the message size).

For the  single-server PIR problem, the only solution to preserve	the information-theoretic demand privacy consists of downloading the whole library. Numerous works have considered the minimization of the communication cost for the system with multiple servers. 
A two-server PIR scheme with communication cost of $O(N^{1/3})$ was proposed in~\cite{chor1995private} based on covering codes \cite{cohen1997covering}, which was then extended to the $S$-server system with   communication cost 
 $O\left(N^{1/(2S-1)}\right)$~\cite{ambainis1997upper}. Also in \cite{itoh1999efficient}, the author introduced a time-efficient two-server PIR with the same communication complexity of $O(N^{1/3})$ as in \cite{chor1995private}. 
	In \cite{beimel2001information} the authors considered the problem of $t$-private PIR where the goal is to keep the identity of the demanded file private, even with the collusion of up to $t$ servers. 
	Some other important works on PIR include~\cite{razborov2006omega, wehner2005improved, chakrabarti2007nearly, beigel2003nearly}, where the authors study the bounds on communication cost. 
	The work in \cite{woodruff2005geometric} is a polynomial-based approach and reaches the $O(N^{1/3})$ communication cost and the work in \cite{dvir20162}, introduces the best known communication cost of $N^{o(1)}$ for two-server PIR schemes as of today and is based on the polynomial approach of \cite{woodruff2005geometric} and matching vector codes (MVC) \cite{efremenko20123, yekhanin2008towards}. 
	
Due to the difficulty to characterize the optimal communication cost, another direction on the PIR problem is to characterize the optimal total download cost. 
	In \cite{sun2017capacity} Sun and Jafar characterized the optimal total download cost,  $1+1/S+1/S^2+\dots+1/S^{N-1}$,  by  proposing  an interference alignment-type achievable scheme and a matching converse.  
In \cite{tian2019capacity}, the authors introduce an asymmetric PIR scheme which achieves the optimal total download cost. Furthermore, under the constraint of achieving the optimal total download cost, this scheme has the minimum total upload cost $S(N-1)\log_2 S$ and the minimum subpacketization on the message $S-1$. 
	In addition, some extended PIR models were considered with the objective to minimize the total download cost, including multi-message PIR (where  the user wants to privately retrieve  $M$ messages from the servers)~\cite{banawan2018multi}, symmetric PIR (where there is an additional security constraint that the user cannot  receive any information about the   undesired messages)~\cite{sun2018capacity}, PIR with side information (where the user has some prior side information in the form of a subset of messages not including the desired one) \cite{kadhe2019private}, PIR from MDS-coded data in distributed storage systems \cite{tajeddine2018private}, cache-aided PIR (where the user has a cache storage that can be used to store any function of the messages) \cite{tandon2017capacity}, multi-message PIR with private side information (where the identity of the desired messages and the side information should be kept private from the servers) \cite{shariatpanahi2018multi}, PIR with colluding databases (where   a number of databases may share the received queries among each other) \cite{sun2017collcapacity}, and PIR with coded databases \cite{banawan2018capacity}.

	\subsection{Contributions}
In the coded caching problem with private demands~\cite{wan2020coded}, an important condition to design non-trivial private caching schemes is that the cache information of each user is private to the others; otherwise, to preserve the demand privacy we need to let each user recover the whole library. However, in most existing private caching schemes (except the virtual users scheme in~\cite{DBLP:journals/corr/abs-1909-03324}), the users' caches are leaked after the transmission in the delivery phase; thus after one  transmission  round where each user has recovered one file, these schemes cannot be used to preserve the demand privacy when each user wants to retrieve another file in a new transmission round.  This motivates the formulation of
	 the coded caching problem with private demands and caches in this paper:  in addition to the privacy constraint on the users' demands, we also want to preserve the privacy of the users' caches. 
Besides the formulation of this new problem, 	our contributions are as follows.
	\begin{itemize}
		\item We first show that the virtual users scheme in \cite{DBLP:journals/corr/abs-1909-03324} is private in terms of demands and caches. The achieved load of this scheme is order optimal within a constant factor except for the case where $N> K$ and $M<N/K$. However, the subpacketization of this scheme is $2^{\mathcal{H}(M/N) NK }$ and is at most  exponential to  $NK$, where $\mathcal{H}(\cdot)$ represents the binary
entropy function.
		\item 
		In order to reduce the subpacketization of the  virtual users scheme, we propose a new construction structure on private coded caching schemes by leveraging two-server PIR schemes.  
		In particular, we show that the schemes resulting from our construction are demand private. By applying
		 the PIR scheme in~\cite{tian2019capacity},
we can construct a demand-private coded caching scheme with an improved memory-load tradeoff than that of~\cite{yan2021fundamental}. 
		We then show that if the underlying PIR scheme has
the uniform demand and  independent query (UDIQ) property (see Definition~\ref{def:udiq} in Section~\ref{sub:review of PIR}), the resulting caching scheme is both demand and cache private. This first results introduce a new ``structural’’ result which constructs the coded caching schemes from a particular class of  PIR schemes.
		\item As a consequence of the above result, we then shift our focus to the construction of two-server PIR schemes with the UDIQ property.  Interestingly, 
we find a new construction structure on two-server PIR schemes under the UDIQ condition by leveraging coded caching schemes. By applying the Maddah-Ali and Niesen scheme into our construction, the achieved load by the resulting two-server PIR scheme 	is proved to be order optimal under the constraint of UDIQ. 
		This is a second new ``structural’’ result, somehow closing the loop in the relation between coded caching and PIR.

		\item In order to explore a broader tradeoff between subpacketization order, transmission load, and cache privacy, we relax the UDIQ constraint,   and as a result, obtain demand private coded caching schemes with a controlled amount of leakage on the cache information, which opens the path in this new exploration. In particular, using the PIR scheme in~\cite{dvir20162}, we obtain a demand private coded caching scheme with better cache informaiton leakage than~\cite{yan2021fundamental}. Recall that using the PIR scheme in~\cite{tian2019capacity}, we obtain a demand private coded caching scheme achieving load lower than~\cite{yan2021fundamental} with the same subpacketization. These results clearly show the flexibility of our construction. 
		
	\end{itemize}

	\subsection{Paper organization}
	The rest of this paper is organized as follows. The system model is presented in 
	 Section \ref{sec:sysmodel}.  Section \ref{sec:results} presents   our main results on coded caching with private demands and caches. Section \ref{sec:extension} 
	 presents the results for the extended model where some leakage on the caches is allowed. 
	We conclude the paper in Section \ref{sec:conclusion}.
	
	\subsection{Notation convention}
	Calligraphic symbols denote sets, bold symbols denote vectors, and sans-serif symbols denote system parameters. We denote the set $\{a,a+1,\dots,b\}$ by $[a:b]$ and $[b]$ refers to $[1:b]$. We use $|\cdot|$ to denote the cardinality of a set or the length of a vector. Also $B_{\mathcal{A}}$ denotes the set $\{B_i, \forall i \in \mathcal{A}\}$. The base of logarithm in this paper is $2$.
	
	\section{System Model} \label{sec:sysmodel}
	\subsection{Problem formulation of coded caching with private demands and caches}
	\label{sub:problem setting}
	The considered coded caching system consists of a server with access to a library of $N$ independent files denoted by $W_1,W_2,\ldots, W_N$.   This server is connected to $K$ cache-aided users with a shared link. The entropy of the cache content of each user is limited by $MF$. We assume that each  file has $F$ bits. The system operates in two phases.
	
	{\it Placement Phase.}
	Each user fills its cache without knowledge of later demands. The cached content of user $k\in[K]$ is
	\begin{align}
		Z_{k}=\phi_{k}(W_{1},\ldots,W_{N}, \mathscr{M}_k), \label{eq:cK}
	\end{align}  
	where   
	$\mathscr{M}_k$ represents the metadata of the  bits in  $Z_k$. 	$\mathscr{M}_k$  is a  random variable over $\mathscr{C}_k$, representing all types of cache placements of user $k$. 
	The realization of $\mathscr{M}_k$ is only known by the server and user $k$. The memory size constraint states that the cache size should be $MF$,\footnotetext{The number of bits per information file symbol 
	to represent cache $Z_k = z_k$ is 
	$
		(1/F) * \lceil \log_2( 1/ \Pr(Z_k = z_k)) \rceil.
	$
	Thus the average number of bits to represent the cache is
	$
		(1/F) \sum_{z_k} \lceil \log_2( 1/ \Pr(Z_k = z_k))\rceil * \Pr(Z_k = z_k) \simeq H(Z_k)/F,
	$
	where the error (rounding integer) is $O(1/F)$. 
	 Hence, for large $F$, we can neglect such rounding, and impose a constraint on the cache entropy, $H(Z_k) \leq MF$. This does not mean that every realization of $Z_k$ can be represented with $MF$ bits, however, on average over the ensemble of the realizations; this holds for each user $k$. } i.e., 
	\begin{align}
	H(Z_{k}) \leq   MF, \ \forall k \in [K].   \text{\footnotemark} \label{eq:memory size}
	\end{align}

Following the assumption made in~\cite{wan2020coded}, we assume that $F$ is sufficiently large such that the size of $\mathscr{M}_k$ is negligible with respect to the file size and $\mathscr{M}_k$ is also provided in $Z_k$. 
	
	
	{\it Delivery Phase.}
	During the delivery phase, user $k \in [K]$ requests one file $W_{d_k}$, where $d_k$ is   uniformly i.i.d. over $[N]$. The demanded vector is denoted by $\mathbf{d}=(d_1,d_2,\ldots,d_{K})$. Given the demand vector $\mathbf{d}$, the server broadcasts to all users  
	\begin{align}
		X_\mathbf{d} = \psi(\mathbf{d}, W_1, \ldots, W_N, \mathscr{M}_1, \ldots,  \mathscr{M}_K). \label{eq:psi}
	\end{align}
	Note that we have  
	\begin{align}
		H(W_{[N]},  \mathscr{M}_{[K]}, \mathbf{d} )= NF + H( \mathscr{M}_{[K]}) + \sum_{k\in [K]} H(d_k) . 
		\label{eq:independent}
	\end{align}
We also assume that the metadata 	of the broadcast message is given inside the message and is negligible compared to the file size.

	{\it Decoding.} 
	User $k\in [K]$ decodes its desired file $W_{d_k}$ from $\big(d_k, Z_k, X_\mathbf{d} \big)$, i.e.,  
	\begin{align}
		H\big(W_{d_k} | d_k, Z_k, X_\mathbf{d}  \big)=0. \label{eq:decoding}
	\end{align}
	
	{\it Privacy.} 
	We want to preserve the privacy of each user's demand against other users, i.e.,
	\begin{align}
		I(\mathbf{d}; X_\mathbf{d} | d_k, Z_k)=0, \ \forall k\in [K].\label{eq:demand privacy constraint} 
	\end{align}
	In addition to~\eqref{eq:demand privacy constraint}, we want to preserve the privacy of the metadata of each user's cache against other users, i.e.,
	\begin{align}
		I\big( (\mathscr{M}_1,\ldots,\mathscr{M}_{K}); X_\mathbf{d} | d_k, Z_k\big)=0, \ \forall k\in [K].\label{eq:cache privacy constraint} 
	\end{align}
	
	{\it Objective.}
	The load $R$ is achievable if there exist  cache placement functions $\{\phi_{k}(\cdot):k\in [K]\}$, encoding function $\psi(\cdot)$, and decoding functions $\{\theta_k(\cdot):k\in [K]\}$ such that
	\begin{align}
		&W_{d_k} = \theta_k(d_k, Z_k, X_\mathbf{d}), \forall k\in [K], \\
		&\text{where } H(X_\mathbf{d})/F \leq R .
	\end{align}
	
	Our objective is to find the minimum achievable load $R^{\star}$ for given system parameters $M,N,K$,\footnote{\label{foot:the same load}Note that the broadcast messages for different demands have the same size, by the constraint of private demands.}  i.e.,
	\begin{align}
		R^{\star}= \min_{\phi_{k}, \psi, \theta_k : k\in [K]}  R. \label{eq:minimum load}
	\end{align}
	
	\subsection{Review of the existing schemes for coded caching with private demands}
	\label{sub:review of private caching}	
	Note that if we remove the private constraint on the caches in~\eqref{eq:cache privacy constraint}, the considered problem reduces to the coded caching problem with private demands in~\cite{wan2020coded}. In the following, we review in brief two efficient existing coded caching schemes with private demands, which are based on the virtual-user strategy and the privacy key strategy, respectively. In the virtual-user strategy proposed in~\cite{kamath2020demand},  an $(N,K,M)$-private scheme is built using an $(N,NK,M)$ non-private scheme. In the placement phase,  user $k$'s cache encoding function, $\text{cache}_k$, in the private scheme is given by $\text{cache}_k=\text{cache}_{(k-1)N+S_k}^{np}$ for $S_k$ chosen uniformly random from $[N]$ where $\text{cache}_i^{np}$ is the cache encoding function of the non-private scheme for user $i$. The memory-load tradeoff in this scheme is given by the piecewise linear function joining the memory-load points $\left(\frac{ t}{K}, \frac{\binom{NK}{t+1}-\binom{NK-N}{t+1} }{\binom{NK}{t}} \right), \ \forall t\in [0:NK]$. Regarding cache privacy, since the choice of $S_k$ is unknown to users other than $k$, the cache contents of the users are private. In the delivery phase, the assignment of files demanded by the virtual users is such that all $N$ file indices are requested by users $[(k-1)N+1:kN]$ for every $k \in [K]$. Regarding privacy, the assignment of caches to demands are revealed during server transmission, but the cache-demand pair for real user $k \in [K]$ is not distinguishable among users $[(k-1)N+1:kN]$ and thus, both caches and demands are private. 
	
	The privacy key scheme proposed in~\cite{yan2021fundamental} does not provide full privacy for users' caches but keeps the demands private. The placement phase is similar to the MAN scheme and the subfiles cached in the MAN scheme for each user is also cached here but also additionally, a linear function of subfiles for each uncached subfile index is stored in the cache for each user. The coefficients of this linear combination are chosen randomly by each user and kept private from others. In the delivery phase, based on these coefficients, the user requests a linear function of subfiles so that the retrieval of the demanded subfiles are possible. If the coefficient vector $\textbf{p}_k$ is used in the placement phase for user $k$, then in the delivery phase the requested coefficient vector would be $\textbf{p}_k+\textbf{d}_k$ in which $\textbf{d}_k$ has a $1$ in position $d_k$ (the demanded file index of user $k$) and $0$ elsewhere. The memory-load tradeoff for this scheme is given by the piecewise linear function joining the memory-load points $\left(1+\frac{t(N-1)}{K}, \frac{\binom{K}{t+1}-\binom{K-\min \{N-1,K\}}{t+1}}{\binom{K}{t}}\right), \ \forall t\in [0:K]$. Regarding privacy, since $\textbf{p}_k$ is a uniformly chosen random vector on $\{0,1\}^N$, $\textbf{p}_k+\textbf{d}_k$ would also be uniformly distributed on $\{0,1\}^N$ no matter the choice of $\textbf{d}_k$ and as a result, the demands are kept private. 
	
	\subsection{Review of   private information retrieval}
	\label{sub:review of PIR}	
	Since our main result is built on a newly discovered connection between private caching schemes and PIR, in this section we review the PIR problem setting. 
	

 Assume that there are $S$ servers each containing of a library of $N$ files with $B$ bits,  denoted by $W_1, W_2, \ldots,W_N$. 
A user is connected to these $S$ servers through $S$ individual and private links (meaning the servers do not collude), and wants to retrieve one file from the library while keeping the privacy of the demand against the servers. 
Assuming that the desired file is $W_{d}$, 
  for each $s\in[S]$, the user sends the query  $Q^{[d]}_{s} \in \mathcal{Q}_{s}$ to server $s$.
	 Based on the received query, server $s$ sends back the answer $A^{[d]}_{s}$ as a function of the query and the files $W_1, W_2, \ldots,W_N$,  to the user; i.e.,
	 \begin{align}
	 	A^{[d]}_{s} = \gamma_{s} (Q^{[d]}_{s}, W_1, W_2, \ldots,W_N),
	 	\label{eq:PIR encoding function}
	 \end{align}
	 where $\gamma_{s}$ represents the encoding function of server $s$. Based on the set of answers and queries, there should exist a decoding function by which the user can recover the desired file, i.e.
	 \begin{align}
	 	H\big(W_d|A^{[d]}_1,\dots,A^{[d]}_S,Q^{[d]}_1,\dots,Q^{[d]}_S\big)=0.
	 	\label{eq:PIR decoding constraint}
	 \end{align}
	 Additionally, the privacy constraint states that the query sent to each server, should not reveal any information about the desired file index; i.e., for each $s \in [S]$, 
	 \begin{align}
	 	I(d;Q^{[d]}_{s}| W_1,\ldots, W_N )=0.
	 	\label{eq:PIR privacy constraint}
	 \end{align}
	 
	 From \eqref{eq:PIR privacy constraint} we can conclude that $H(A^{[1]}_{s})=\cdots=H(A^{[N]}_{s}):=H(A_{s})$ holds for every $s \in [S]$. The total download cost of the PIR scheme is defined as the total size of information received from the servers over message size, denoted by 
	 $R_D=\frac{\sum_{s \in [S]} H(A_{s}) }{B}.$ 
	  The objective of the PIR problem is to characterize the minimum total download cost $R_D$.\footnote{\label{foot:PIR rate}Note that in most information theoretic works on PIR, the objective is to maximize the download rate, which is defined as $\frac{B }{\sum_{s \in [S]} H(A_{s})}$. In other words, the download rate is the reciprocal of the  total download cost considered in this paper. 
	 } 

 	The optimal total download cost was solved for general systems parameters in~\cite{sun2017capacity}. This result is recalled here in the following: 
 	\begin{theorem} [Capacity of PIR~\cite{sun2017capacity}]
 		\label{thm:pircapacity}
 		For the PIR problem with $N$ messages and $S$ databases, the optimal total download cost is 
 		\begin{align}
 			 1+1/S+1/S^2+\dots+1/S^{N-1}.
 		\end{align}
 	\end{theorem}
 	Interestingly, the optimal total download cost can be achieved not only in the asymptotic regime of arbitrarily large file size. In fact, it is sufficient that the file size is equal to any integer multiple of $S^N$ bits. 
 	 In~\cite{tian2019capacity} the authors proposed a PIR scheme   that reaches the optimal (least) file size and  total upload cost among the class of decomposable codes achieving the optimal total download cost. According to~\cite{tian2019capacity}, the term ``decomposable" restricts each coded symbol to be a summation of the component functions on the individual messages. For the exact definition,  please refer to~\cite[Definitions 2 and 3]{tian2019capacity}. Their result is stated in the following theorem. 
    \begin{theorem} [\hspace{1sp}\cite{tian2019capacity}]
 		Among all download cost optimal uniformly decomposable PIR codes, the PIR code proposed in~\cite{tian2019capacity} has the smallest message size, which is $S-1$. Among all download cost optimal decomposable PIR codes, this scheme has the lowest total upload cost, which is $S(N-1) \log_2 N$. 
 		
 	\end{theorem} 
 
 	Finally, we introduce the \textit{uniform demand and independent queries} (UDIQ) condition on PIR schemes which will be needed in our construction of coded caching schemes with private demands and caches. 
 	\begin{definition} [UDIQ condition] \label{def:udiq}
 		For a two-server PIR scheme, if the demand is uniformly distributed over $[N]$ and  
 		\begin{align}
 			I( Q^{[n]}_{1}; Q^{[n]}_2| W_1,\ldots, W_N)=0,  \ \forall n \in[N], \label{eq:independent queries constriants}
 		\end{align}
 		then the PIR scheme  satisfies the UDIQ condition. 
 	\end{definition}

	\section{Main Results} \label{sec:results}
In this section, we will present our main results on the coded caching problem with private demands and caches. 	
	We first show that the virtual users scheme reviewed in Section~\ref{sub:review of private caching} can also preserve the privacy of the users' caches. 
	\begin{theorem}
		\label{thm:virtual user}
		For the coded caching problem with private demands and caches, $R^{\star} $ is upper bounded by the lower convex envelop of the following memory-load tradeoff points, 
		\begin{align}
			\left(\frac{ t}{K}, \frac{\binom{NK}{t+1}-\binom{NK-N}{t+1} }{\binom{NK}{t}} \right), \ \forall t\in [0:NK]. 
			\label{eq:upper bound on load}
		\end{align}
	\end{theorem}
	
	\begin{proof}
		The demand privacy constraint in \eqref{eq:demand privacy constraint} is already proved to hold in \cite{DBLP:journals/corr/abs-1909-03324}. To complete the proof, we need to show that the cache privacy constraint also holds. The virtual users scheme for parameters $(N,K,M)$ is build upon the non-private MAN scheme for parameters $(N,NK,M)$ when the demands of virtual users are carefully selected. In this scheme, user $k$ acts as user $\big((k-1)N+S_k\big)$ in the $(N, NK, M)$ non-private scheme in which $S_k\sim\text{Unif}\{[N]\}$. So the metadata of the cache content of user $k$ is determined by $S_k$, or equivalently $\mathscr{M}_k=S_k$. In this scheme, the users $(k-1)N+1, (k-1)N+2,\ldots,kN$ cover all $N$ possible demands. Following the demand construction of the $(N,NK,M)$ non private scheme in \cite{DBLP:journals/corr/abs-1909-03324}, define $C_k$ as follows, 
		\begin{align} \label{eq:shifts}
			C_k:=(S_k-d_k) \text{ mod } N, k \in [K]
		\end{align}
Then, let $q_k$ be the right cyclic shift of the vector $(1, \ldots, N)$ by $C_k$ positions. Thus the demand vector of the 		
		  $(N,NK,M)$ non-private scheme is $\mathbf{d}^{np}=(\mathbf{q}_1, \mathbf{q}_2, \ldots, \mathbf{q}_K)$. So we can see the demand vector in the non-private scheme is a function of   $\mathbf{C}:=(C_1,C_2,\ldots,C_K)$. The transmission of the server for one part should contain the vector $\mathbf{C}$ in order for the users to be able to decode their messages    \cite{DBLP:journals/corr/abs-1909-03324}. The other part of the transmission consists of a non-private $(N,NK,M)$ coded caching scheme based on the scheme in \cite{yu2017exact} which is a function of the library and $\mathbf{d}^{np}$ and since $\mathbf{d}^{np}$ is a function of $\mathbf{C}$, we denote this part of transmission as $X^{np}(W_{[K]}, \mathbf{C})$. So in the end we can write $X_\mathbf{d}=(\mathbf{C},X^{np}(W_{[K]}, \mathbf{C}))$. Now we can write the cache privacy criterion in \eqref{eq:cache privacy constraint} as follows,  
		\begin{subequations}
			\begin{align}
				&I\big( (\mathscr{M}_1,\ldots,\mathscr{M}_{K}); X_\mathbf{d} | d_k, Z_k\big) \nonumber \\ 
				&= I\big(S_1,\ldots,S_K;\mathbf{C},X^{np}(W_1,\ldots,W_K, \mathbf{C}) | d_k, Z_k\big) \\
				&=I(S_1,\ldots,S_K;\mathbf{C} | d_k, Z_k) \label{1}\\
				&+I(S_1,\ldots,S_K;X^{np}(W_1,\ldots,W_K, \mathbf{C}) | d_k, Z_k, \mathbf{C}) \label{2}
			\end{align}
		\end{subequations}
		Based on \eqref{eq:shifts} and the fact that the demands are uniformly distributed, the distribution of $\mathbf{C}$ does not change depending on knowing or not knowing the value of the vector $(S_1,\ldots,S_K)$. Thus the term in \eqref{1} is zero and since $\mathbf{C}$ is already in the condition in \eqref{2}, $X^{np}(W_1,\ldots,W_K, \mathbf{C})$ would not have any connection to $(S_1,\ldots,S_K)$ and this term is also zero. Therefore, both the privacy constraints \eqref{eq:demand privacy constraint} and \eqref{eq:cache privacy constraint} are satisfied and decodability in \eqref{eq:decoding} is already proved to hold in \cite{DBLP:journals/corr/abs-1909-03324}. This completes the proof. 
	\end{proof}

Note that it was proved in~\cite{wan2020coded} that the multiplicative gap between the achieved load by the virtual users scheme and 
 the converse bound of the non-private coded caching problem is at most $8$, except the case of $N<K$ and $M<N/K$. This order optimality result also holds for the considered coded caching problem with private demands and caches.

	\subsection{New construction on  coded caching with private demands}
	\label{sub:new construction}
		The subpacketization   of the virtual users scheme in Theorem~\ref{thm:virtual user} is     $2^{\mathcal{H}(M/N) NK }$ and is at most  exponential to  $NK$, while the subpacketization of the MAN scheme is  $2^{\mathcal{H}(M/N) K }$ and is at most  exponential to  $K$. Next, we aim   to   reduce the subpacketization   of the virtual users scheme while keeping demand and cache information private simultaneously.  
The key contribution of this paper is to propose a new construction strategy on private coded caching, which 
establishes a new relationship between two-server PIR schemes and private coded caching. 
We first consider demand-privacy, and propose a structure in the following theorem to construct  demand private coded caching schemes from PIR schemes. 
 The proof is given in Appendix~\ref{proofthm4}. 	
 	\begin{theorem}[From PIR to coded caching]
 	\label{thm:private demands}
 	Given any two-server PIR scheme with $N$ files and download cost pair $(R_{D_1}, R_{D_2})$ where $R_{D_i}$ corresponds to server $i$, there exists an $(N,K)$ coded caching scheme ($N$ files and $K$ users) with private demands  whose achieved memory-load tradeoff is the lower convex envelope of $(0,N)$, 
 	 \begin{align}
 		\left(  \frac{Nt}{K}+\big(1-\frac{t}{K}\big) \left(\mu_1 R_{D_1}+ \mu_2 R_{D_2}\right), \left(\mu_1 R_{D_2}+ \mu_2 R_{D_1}\right) \frac{K-t}{t+1}  \right) ,   \forall t\in [0:K-1],
 		\label{eq:novel approach achieved memory load}
 	\end{align}
 	and $(N,0)$, where $\mu_1,\mu_2 \in [0,1], \mu_1+\mu_2=1$. Assume the needed subpacketization of the given PIR scheme  is $F^{\prime}$, then the needed subpacketization for each point in \eqref{eq:novel approach achieved memory load} with $t\in [0:K-1]$ is $\binom{K}{t} F^{\prime}$.  
 	\end{theorem}
 

 
 	Since based on Theorem \ref{thm:private demands} we are allowed to use any two-server PIR scheme, we can choose the one in \cite{tian2019capacity} which has the optimal total download cost $R_D^{\star}:= 1+1/2+(1/2)^2+\cdots+(1/2)^{N-1}$ and subpacketization level of $F'=1$. Therefore, using the scheme in \cite{tian2019capacity} into Theorem \ref{thm:private demands} (for $\mu_1=\mu_2=1/2$), we will have the following result.
 	\begin{corollary}
 		\label{cor:private demands}
 		For the $(N,K)$  coded caching problem with private demands in \cite{wan2020coded}, there exists a scheme whose achieved memory-load tradeoff is the lower convex envelope of $(0,N)$, 
 		\begin{align}
 			(M,R)= \left(  \frac{Nt}{K}+\big(1-\frac{t}{K}\big)\frac{R_D^{\star}}{2}, \frac{R_D^{\star}}{2}\frac{K-t}{t+1}  \right) ,   \forall t\in [0:K-1],
 			\label{eq:achieved memory load private demands}
 		\end{align}
 		and $(N,0)$. The needed subpacketization for each point in~\eqref{eq:achieved memory load private demands} with $t\in [0:K-1]$ is $\binom{K}{t}$.
 	\end{corollary}
 	
 	\begin{remark}[Comparison to the demand private caching scheme in~\cite{yan2021fundamental}] \label{re:simp}
 	
 In Theorem \ref{thm:private demands} for the time-sharing parameters   $\mu_1=\mu_2=1/2$ (or the case where $R_{D_1}=R_{D_2}$), the points in \eqref{eq:novel approach achieved memory load} become
 		\begin{align}
 			(M,R) = \left(\frac{Nt}{K}+\big(1-\frac{t}{K}\big) \frac{R_D}{2}, \frac{R_D}{2} \frac{K-t}{t+1}\right), \ \forall t \in [0:K], \label{eq:half half}
 		\end{align}
 	where $R_D=R_{D_1}+R_{D_2}$ represents the total download cost.
 	The memory-load tradeoff for the demand private scheme of \cite{yan2021fundamental} for the case $K \leq N+t$, follows  $(M,R)= \left(  \frac{Nt}{K}+\big(1-\frac{t}{K}\big), \frac{K-t}{t+1}  \right)$, which is  order optimal within a constant gap. So when $K \leq N+t$,  the achieved memory-load tradeoff in~\eqref{eq:half half} is strictly better than    \cite{yan2021fundamental} if the selected PIR scheme has the total download cost   $ R_D=R_{D_1}+   R_{D_2}<2$. In this case, the resulting scheme is also order optimal within a constant gap. 
 	 When $K \leq N+t$, the demand private coded caching scheme in \cite{yan2021fundamental} is a special case of our construction in Theorem~\ref{thm:private demands}; by applying the two-server PIR scheme in~\cite{shah2014one} into Theorem~\ref{thm:private demands}, the resulting coded caching scheme with private demands becomes  the privacy key scheme in \cite{yan2021fundamental}. Since the total download cost of the two-server PIR scheme in~\cite{tian2019capacity} is strictly lower than $2$,  when  $K \leq N+t$ the proposed caching scheme   in Corollary \ref{cor:private demands} has a strictly better performance on the memory-load tradeoff than the scheme in \cite{yan2021fundamental}, while the needed subpacketizations of these two schemes are the same.   Note that when $K > N+t$, the proposed demand-private scheme is also order optimal within a constant gap, by using a similar proof as~\cite[Appendix D]{yan2021fundamental}.\footnote{\label{foot:order optimal for K>N}More precisely, by the same proof for the case $M\leq 1$, we can show the load equal to $N$ is order optimal within a factor of  $4$; when $M>1$, we can show that the gap between the proposed scheme and the MAN scheme is within a constant gap. In addition, the memory-sharing between $(0,N)$ and the MAN scheme is order optimal within a factor of $4$~\cite{improvedlower2017Ghasemi}. So we can prove that our scheme is also order optimal within a constant gap.}
 	\end{remark}


 	\begin{remark}[Comparison to the virtual users scheme]
 		A comparison on the loads of the virtual users scheme in Theorem \ref{thm:virtual user} and our construction with the optimal PIR scheme in Corollary \ref{cor:private demands} is depicted in Figure \ref{fig:compare}. Note that the subpacketization of the virtual users scheme is   $2^{\mathcal{H}(M/N) NK }$ and is at most  exponential to  $NK$, while that of Corollary \ref{cor:private demands} is $2^{\mathcal{H}(M/N) K }$ and is at most  exponential to  $K$.
 	\end{remark}

 	\begin{figure} [h]
 		\centering
 		\includegraphics[width=120mm]{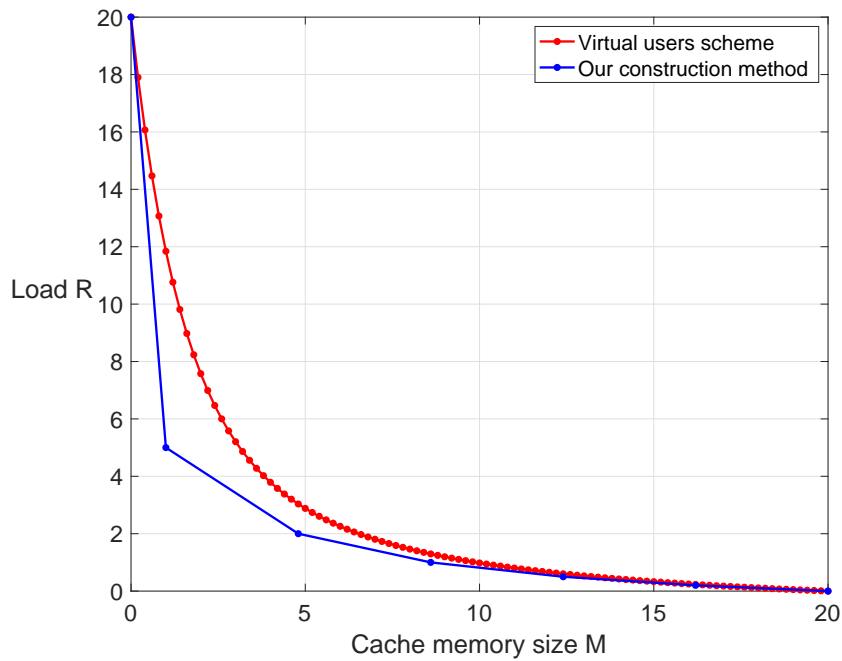}
 		\caption{Comparison of loads for parameters $N=20, K=5$, and different values of $M$ for the virtual users scheme in Theorem \ref{thm:virtual user} and our construction in Corollary \ref{cor:private demands}. }
   \label{fig:compare}
 	\end{figure}

 	\begin{remark}
 		The connection of PIR and demand private coded caching in our structure in Theorem~\ref{thm:private demands}, emerges from the fact that the individual queries sent to the servers solely do not reveal any information about the demanded file index. Therefore, the query to one server can be used to fill out the cache memory and the query to the other server to build up the server transmission, without revealing any information about the demanded indices by the users. This logic holds for any PIR scheme including multi-message PIR schemes. Specifically, if we assume $d_k, k \in [K]$ denotes the set of demands by user $k$, and $\mathbf{d} = (d_1, d_2, \dots. d_K)$, the proof of Theorem \ref{thm:private demands} in Appendix \ref{proofthm4} works without any change. In this case by using these schemes, each user can request multiple files in the coded caching scheme while preserving the privacy of these demands. 
 	\end{remark}

 We can also  extend the proposed construction in Theorem~\ref{thm:private demands} to obtain a more flexible tradeoff among the memory, load, and subpacketization, 
 by using any  coded caching scheme under PDA construction~\cite{yan2017placement}, instead of the MAN caching scheme (recall that the MAN scheme can be also seen as a caching scheme under PDA construction). This extension is feasible because the coded caching schemes under PDA construction  is based on uncoded cache placement (which is symmetric across files) and clique-covering delivery.\footnote{\label{foot:clique covering}The clique-covering delivery means that, in the delivery phase several multicast messages are broadcasted to the users. Each multicast message is a sum of subfiles and  useful to a subset of users, where each user requests one subfile and caches all the other subfiles.} Directly from the proof of Theorem~\ref{thm:private demands}, we can have the following corollary.
 	\begin{corollary}
 		\label{cor:private demands with PDA}
 		Given any two-server PIR scheme with $N$ files and download cost pair $(R_{D_1}, R_{D_2})$, and given any non-private coded caching scheme under PDA construction with memory-load tradeoff $(M_1,R_1)$, 
  	there exists an $(N,K)$ coded caching scheme   with private demands  which can  achieve the memory-load tradeoff point
 	 \begin{align}
 		(M,R)= \big( M_1+ (1-M_1/N) \left(\mu_1 R_{D_1}+ \mu_2 R_{D_2}\right) ,\left(\mu_1 R_{D_2}+ \mu_2 R_{D_1}\right) R_1  \big), 
 		\label{eq:novel approach achieved memory load with PDA}
 	\end{align}
 	where $\mu_1,\mu_2 \in [0,1], \mu_1+\mu_2=1$. Assume the   subpacketizations of the given PIR scheme   and of the non-private coded caching scheme are   $F^{\prime}$ and $F^{\prime\prime}$, respectively;   then the needed subpacketization  of the resulting coded caching scheme   with private demands  is $  F^{\prime} F^{\prime\prime}$.  
 	\end{corollary}
By applying coded caching schemes under PDA construction into Corollary~\ref{cor:private demands with PDA}, we can further reduce the subpacketization of the scheme in Theorem~\ref{thm:private demands}.

\subsection{New construction on  coded caching with private demands and caches}
Next, we consider the construction of coded caching schemes with both demand privacy and cache privacy. This is given in the following result, proved in Appendix~\ref{proofthm5}.
 	\begin{theorem}  \label{thm:cachepir}
 		Given any two-server $N$-message PIR scheme satisfying the UDIQ condition in Definition~\ref{def:udiq} with download cost pair $(R_{D_1}, R_{D_2})$ where $R_{D_i}$ corresponds to server $i$ and time-sharing parameters $\mu_1,\mu_2$ where $\mu_1,\mu_2 \in [0,1], \mu_1+\mu_2=1$, there exists an $(N,K)$  coded caching scheme with private demands and caches whose achieved memory-load tradeoff  $(M,R)$ is the lower convex envelope of $(0,N)$, $(N,0)$, and the points in~\eqref{eq:novel approach achieved memory load}. Assume the needed subpacketization of the given PIR scheme  is $F^{\prime}$, then the needed subpacketization for each point in \eqref{eq:novel approach achieved memory load} with $t\in [0:K-1]$ is $\binom{K}{t} F^{\prime}$.  
	\end{theorem} 
 	The novelty in the construction in Theorem \ref{thm:private demands} is to generate private keys by a two-server PIR scheme. In the privacy key scheme~\cite{yan2021fundamental}, in addition to caching subfiles as in the MAN caching scheme, for each  set $\mathcal{V}\subseteq [K]$ where $k\notin \mathcal{V}$ and $|\mathcal{V}|=t$, each user $k$ also caches a random linear combination of $W_{1,\mathcal{V}}, \ldots,W_{N,\mathcal{V}}$ (assumed to be $p_1W_{1,\mathcal{V}}+\cdots+p_{N}W_{N,\mathcal{V}}$) in its caches as a private key, such that the effective demand of user $k$ in the delivery phase becomes 
 	$$p_1W_{1,\mathcal{V}}+\cdots+ p_{d_k-1} W_{d_k-1,\mathcal{V}}+ (p_{d_k}+1)W_{d_k,\mathcal{V}} +p_{d_k+1} W_{d_k+1,\mathcal{V}} +\cdots +p_{N}W_{N,\mathcal{V}}.$$
  	Thus the privacy of  the user's demand   could be preserved. 
  	 In our construction, instead of storing a random linear combination of $W_{1,\mathcal{V}}, \ldots,W_{N,\mathcal{V}}$, we apply any two-server PIR scheme where we treat each of  $W_{1,\mathcal{V}}, \ldots,W_{N,\mathcal{V}}$ as a file in the PIR problem. The answer of the first server in the PIR scheme serves as the     private key  stored by user $k$; according to the demand of user $k$, the answer of the second server in the PIR scheme serves as the effective request of user $k$. 
Then in Theorem \ref{thm:cachepir}, if the PIR scheme additionally satisfies the UDIQ condition, the resulting coded caching scheme satisfies the cache privacy condition in addition to the demand privacy condition.

From the same reason on deriving Corollary~\ref{cor:private demands with PDA}, we can also extend Theorem~\ref{thm:cachepir} by using other coded caching schemes under PDA construction, and obtain the following corollary.
\begin{corollary}  \label{cor:cachepir PDA}
 		Given any two-server $N$-message PIR scheme satisfying the UDIQ condition in Definition~\ref{def:udiq}  with download cost pair $(R_{D_1}, R_{D_2})$, and given any non-private coded caching scheme under PDA construction with memory-load tradeoff $(M_1,R_1)$, 
 		 there exists an $(N,K)$  coded caching scheme with private demands and caches which achieves the memory-load tradeoff the point in~\eqref{eq:novel approach achieved memory load with PDA}. Assume the   subpacketizations of the given PIR scheme   and of the non-private coded caching scheme are   $F^{\prime}$ and $F^{\prime\prime}$, respectively;   then the needed subpacketization  of the resulting coded caching scheme   with private demands  is $  F^{\prime} F^{\prime\prime}$.  
	\end{corollary}

\subsection{New construction on two-server PIR schemes}
\label{sub:two server PIR}

By the proposed construction   in Theorem~\ref{thm:cachepir} (resp. the one in Theorem~\ref{thm:private demands}), in order to design coded caching schemes with private demands and caches (resp. with private demands), our task is to design two-server PIR schemes under (resp. without) the UDIQ condition with total download cost and subpacketization level as low as possible. 
In the following we propose a new construction structure on two-server PIR schemes under the UDIQ condition by leveraging coded caching schemes. 
Intuitively, this idea stems from the observation that, the placement phase of coded caching, does not reveal any information on the demands; and the observation that given the transmission of the delivery phase, from 
different cache configurations we can decode different files. Hence, we can treat the cache configuration of one user as the transmission of one server in the PIR scheme and treat the delivery phase as the transmission of the other server in the PIR scheme.  From the above explanation, we have the following construction.


	\begin{theorem}[From coded caching to PIR]  \label{thm:cacheinpirGen}
	Assume that there exists a coded caching scheme  for $N$ users and $N$ files  which achieves the memory-load trafeoff $(M,R)$ with subpacketization $F$. Then there exists a   two-server $N$-message PIR scheme satisfying the UDIQ condition in Definition \ref{def:udiq} with the download cost pair $(R_{D_1}, R_{D_2})=(M,R)$ and subpacketization $F$.	
	\end{theorem}

\begin{proof}
	We	consider  the coded caching scheme for the   shared-link setting with $N$ files and $K=N$ users. In the cache placement phase,  each user $i\in [N]$ fills its cache by the content denoted by $Z_i$.
	In the delivery phase, each user requests a distinct file. Thus the demand vector $\mathbf{d} = (d_1, \ldots, d_N)$ is a permutation function $\pi(.)$ from $[N]$ to $[N]$. In the delivery phase, the server sends the message $X_{\mathbf{d}}$. By the decodability of the coded caching scheme, from $X_{\mathbf{d}}$ and $Z_i$, we can decode $W_{d_i}$, for each $i\in [N]$. 
	
	Next  we use the above coded caching scheme to construct a two-server PIR scheme under the UDIQ condition. Let us go back to the PIR setting, where the user  requests  file $W_{\theta}$ where $\theta$ is  distributed uniformly at random on $[N]$. 
	The user generates a random variable $r$ uniformly on $[N]$ and sends $r$ as the query to  the first server, in order to retrieve   $Z_r$. In addition, to determine the demand vector, we first define $\mathbf{d}_c$ as $\mathbf{d}_c = (1,2,...,N)$. The demand vector $\mathbf{d}$ is determined as the cyclic shift of $\mathbf{d}_c$ by $<r-\theta>_{N}$ positions to the right; i.e. $\mathbf{d}(i) = \mathbf{d}_c(<i-<r-\theta>_{N}>_{N})$.\footnote{\label{foot:mod}In this paper, we let $<\cdot>_a$ represent  the modulo operation with  integer quotient $a>0$ and  we let $<\cdot>_a \in \{1,\ldots,a \}$ (i.e., we let $<b>_a=a$ if $a$ divides $b$).} Now the user sends $<r-\theta>_{N}$ as the query to the second server to retrieve $X_{\mathbf{d}}$. 
	
	Obviously, the query to the first server is independent of the demand. In addition,	since $r$ is  generated  independently and uniformly, the second server cannot get any information about $\theta$. 
	So the privacy constraint in PIR in \eqref{eq:PIR privacy constraint} is satisfied. On the other hand, since $I(r;<r-\theta>_{N}|W_1,...,W_N) = I(r;\theta)=0$, the UDIQ condition in \eqref{def:udiq} is also satisfied by this scheme. 	
\end{proof}

We then apply the MAN coded scheme with memory-load tradeoff points $(M,R)=(t,\frac{N-t}{t+1})$ and subpacketization level  $\binom{N}{t}$, for $t\in [0:N]$, into the construction in Theorem \ref{thm:cacheinpirGen}. 
	\begin{theorem}\label{thm:cachinginPIR}
	There exists a two-server PIR scheme satisfying the UDIQ condition in Definition \ref{def:udiq}, whose achieved download cost pair  is the convex envelope  of the points $(R_{D_1}, R_{D_2}) = (t, \frac{N-t}{t+1})$
	with subpacketization level   $\binom{N}{t}$, for all $t\in [N]$. By letting $t=O\left(\sqrt{N}\right)$, the resulting two-server PIR scheme  achieves  the download costs $R_{D_1}$ and $R_{D_2}$  of order $O\left(\sqrt{N}\right)$ with   subpacketization level $O\left(\sqrt{N}^{\sqrt{N}}\right)$ (considering highest order in the exponent).
	\end{theorem}

	\begin{remark}
		In this paper we    exploit  the connection between PIR and coded caching, where we use one to build the other one as illustrated in Fig.~\ref{fig:diagram}. More precisely, in Theorem \ref{thm:cachepir} (resp. Theorem \ref{thm:private demands}) we propose a construction structure on demand and cache private (resp. demand private) caching schemes utilizing  two-server PIR schemes satisfying (resp. not satisfying) the UDIQ condition.   Later in Theorem \ref{thm:cacheinpirGen} we propose  a construction structure on  two-server PIR schemes satisfying the UDIQ condition utilizing  coded caching. 	
	\end{remark}
		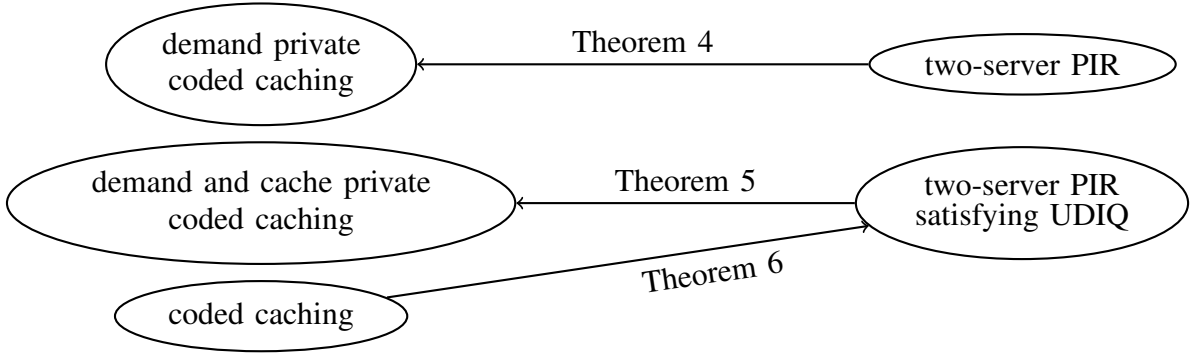
\begin{figure}
		\begin{tikzpicture} [thick, main/.style = {ellipse, draw}] 
			\node[main] (1) {\shortstack{demand private \\ coded caching}};
			\node[main] (2) [right=6cm of 1] {two-server PIR};
			\node[main] (3) [below=2mm of 1] {\shortstack{demand and cache private \\ coded caching}};
			\node[main] (4) [right=4.5cm of 3] {\shortstack{two-server PIR \\ satisfying UDIQ}};
			\node[main] (5) [below=2mm of 3] {coded caching};
			
			\draw[->] (2) -- node[midway, above] {Theorem \ref{thm:private demands}} (1); 
			\draw[->] (4) -- node[midway, above] {Theorem \ref{thm:cachepir}} (3); 
			\draw[->] (5) -- node[midway, below right, sloped] {Theorem \ref{thm:cacheinpirGen}} (4); 
		\end{tikzpicture}
	\caption{Diagram of the proposed connections between PIR and coded caching.}
	\label{fig:diagram}
	\end{figure}

Next, we derive a lower bound on the download costs  of a two-server PIR scheme satisfying the UDIQ condition in Definition \ref{def:udiq} by using a cut-set argument. 
 We assume that the sets of queries to server $1$ and server $2$ are    $\mathcal{Q}_1$ and   $\mathcal{Q}_2$, respectively.  Consider the set of pairs of queries that based on the design of the PIR scheme can be sent to recover file $W_{\tau}$; we denote this set by $\mathcal{U}_\tau$ as follows, 
	\begin{align}
		\mathcal{U}_\tau \triangleq \{ (Q_1, Q_2): Q_1 \in \mathcal{Q}_1, Q_2 \in \mathcal{Q}_2, \text{ $(Q_1, Q_2)$ recovers $W_\tau$} \}.
	\end{align} 
	For a particular choice of $q_1 \in \mathcal{Q}_1$, we define the set of all queries in the set $\mathcal{Q}_2$ that can together recover file $W_\tau$ as follows.
	\begin{align} \label{def:conditional}
		\mathcal{U}_{\tau|Q_1=q_1} \triangleq \{ Q_2 : Q_2 \in \mathcal{Q}_2, (q_1,Q_2) \text{ recovers } W_\tau \}.
	\end{align}
	Similarly, we define 
	\begin{align} \label{def:conditional2}
		\mathcal{U}_{\tau|Q_2=q_2} \triangleq \{ Q_1 : Q_1 \in \mathcal{Q}_1, (Q_1,q_2) \text{ recovers } W_\tau \}.
	\end{align}
	We propose the following converse bound, whose proof could be found in Appendix~\ref{sub:proof of converse}. 
	\begin{theorem}  \label{thm:lowerbound}
		In a two-server PIR scheme satisfying the UDIQ condition in Definition \ref{def:udiq},  denote the query sets to servers 1 and 2 respectively by $\mathcal{Q}_1$ and $\mathcal{Q}_2$, where $|\mathcal{Q}_1|=N_1$ and $|\mathcal{Q}_2|=N_2$; denote the download costs from servers $1$ and $2$  by $R_{D_1}$ and $R_{D_2}$, respectively. If we have uniform query distribution for both servers; $\Pr(Q_1=q_1 \in \mathcal{Q}_1)=1/N_1$ and $\Pr(Q_2=q_2 \in \mathcal{Q}_2)=1/N_2$,\footnote{\label{foot:}{The two-server PIR schemes in Theorem \ref{thm:pirschemes}, satisfy the uniform query distribution condition stated in Theorem \ref{thm:lowerbound}. To our best knowledge, existing information theoretic PIR schemes    also satisfy this condition.}} then, 
		\begin{enumerate}
			\item for all $q_1 \in \mathcal{Q}_1$ and all $q_2 \in \mathcal{Q}_2$, we have 
			\begin{align}
				n_2 \triangleq \left|\mathcal{U}_{\tau|Q_1=q_1}\right|, n_1 \triangleq \left|\mathcal{U}_{\tau|Q_2=q_2}\right|, \ \forall 	\tau \in [N]; 		
			\end{align}
			\item   we have
			\begin{align}
				\frac{N_1}{n_1} \leq N, \frac{N_2}{n_2} \leq N; 
			\end{align}
			
			\item we have
			\begin{align} \label{thm:r1r2lower}
				{\min_{%
					\substack{%
						\alpha_1 \in [N_1], \alpha_2 \in [N_2], \\
						\alpha_1 \alpha_2=\left\lceil\frac{N_1}{n_1}\right\rceil=\left\lceil\frac{N_2}{n_2}\right\rceil
					}
				}}
				\alpha_1 R_{D_1} + \alpha_2 R_{D_2} \geq N
			\end{align}
			\item if we assume $R_{D_1}=R_{D_2}=R'_{D}$ we have
			\begin{align}
				R'_D \geq \frac{N}{2\left(\sqrt{\left\lceil\frac{N_1}{n_1}\right\rceil}+1\right)} = \frac{N}{2\left(\sqrt{\left\lceil\frac{N_2}{n_2}\right\rceil}+1\right)} \geq \frac{N}{2\left(\sqrt{N}+1\right)}.
				\label{eq:sqrt converse}
			\end{align}
		\end{enumerate}
	\end{theorem}
Note that for any two-server PIR scheme, by using time-sharing we can always obtain another two-server PIR scheme with equal download costs from the two servers, where the total download cost is the same as the previous two-server PIR scheme. Hence, it can be seen from~\eqref{eq:sqrt converse} that any two-server PIR scheme satisfying the UDIQ condition in Definition \ref{def:udiq} should have a total download cost 
                \begin{align}
				R_D \geq \frac{N}{\left(\sqrt{\left\lceil\frac{N_1}{n_1}\right\rceil}+1\right)} = \frac{N}{\left(\sqrt{\left\lceil\frac{N_2}{n_2}\right\rceil}+1\right)} \geq \frac{N}{\left(\sqrt{N}+1\right)}=O\left(\sqrt{N}\right).
				\label{eq:sqrt total converse}
			\end{align}
Comparing the converse bound in~\eqref{eq:sqrt total converse} and the proposed two-server PIR scheme in Theorem~\ref{thm:cachinginPIR}, we can obtain the following order optimality result.
	\begin{corollary}
 \label{cor:order optimality}
 The total download cost by the two-server PIR scheme in Theorem~\ref{thm:cachinginPIR}, which is equal to $O\left(\sqrt{N}\right)$, is order optimal under the constraint of  UDIQ  and uniform query.  
 \end{corollary}

	 For some special cases, more precisely for $N\in \{2,3,4\}$, in Appendix~\ref{sub:newPIRschemes} we propose new two-server PIR schemes satisfying the UDIQ condition, whose subpacketizations are lower and the download costs are lower or equal compared to the two-server PIR scheme in Theorem~\ref{thm:cachinginPIR}.
	 	 \begin{theorem}  \label{thm:pirschemes}
		For the two-server PIR schemes satisfying the UDIQ condition in Definition \ref{def:udiq}, 
		
	 	1) when $N=2$, the download cost pair {$(R_{D_1}, R_{D_2})=(0.5,1)$} (i.e., $R_D=3/2$) is achievable and the required subpacketization is $F'=1$;
		
		2) when  $N=3$, the download cost pair {$(R_{D_1}, R_{D_2})=(1,1)$} (i.e., $R_D=2$) is achievable and the required subpacketization is $F'=1$;
		
		3) when $N=4$, the download cost pair {$(R_{D_1}, R_{D_2})=(1,1)$} (i.e, $R_D=2$) is achievable and the required subpacketization is $F'=1$.
		
		
	\end{theorem}
 
Based on Theorems \ref{thm:pirschemes} and \ref{thm:lowerbound}, we readily get the following result.
	\begin{corollary} \label{cor:n24optimality}
		The PIR schemes in Theorem \ref{thm:pirschemes} for the cases $N=2$ and $N=4$, meet the lower bound \eqref{thm:r1r2lower} in Theorem \ref{thm:lowerbound} with equality.
	\end{corollary}
	
	\begin{proof}
		For the case $N=2$, as we mention in Appendix \ref{sec:schemes}, we use the PIR scheme proposed in  \cite[Section III-A]{tian2019capacity}. Remember this scheme has uniform distribution on queries. In this scheme $N_1=N_2=2$, $n_1=n_2=1$, and $R_{D_1}=0.5, R_{D_2}=1$. For the minimization in the left hand side of \eqref{thm:r1r2lower}, we have $\alpha_1 \alpha_2 = \alpha'=2$. The minimum happens when $\alpha_1=2, \alpha_2=1$. Then,
		\begin{align}
			2R_{D_1}+R_{D_2}= 2 = N.
		\end{align}
		So this case holds \eqref{thm:r1r2lower} with equality.
		
		
		For the case $N=4$ introduced in Appendix \ref{sec:schemen4}, we have $N_1=N_2=4$, $n_1=n_2=1$, and $R_{D_1}=1, R_{D_2}=1$. Again remember, this scheme has uniform distribution on queries. For the minimization in the left hand side of \eqref{thm:r1r2lower}, we have $\alpha_1 \alpha_2 = \alpha'=4$. The minimum happens when $\alpha_1=2, \alpha_2=2$. Then,
		\begin{align}
			2R_{D_1}+2R_{D_2} = 4 = N. 
		\end{align}
		So this case also holds \eqref{thm:r1r2lower} with equality.
	\end{proof}

By applying the proposed two-server PIR schemes  in Theorems~\ref{thm:cachinginPIR} and~\ref{thm:pirschemes}  into our construction in Theorem~\ref{thm:cachepir}, we can directly obtain the following coded caching schemes with private demands and caches. Note that for the first three parts, we use the schemes of Theorem \ref{thm:pirschemes}, and for the last part, we use the scheme in Theorem \ref{thm:cachinginPIR}.
	\begin{corollary}
		\label{cor:tradeoffdiffN}
		For the coded caching problem with private demands and caches, we have the following achievable schemes:
		
		1) when $N=2$, the following memory-load points are achievable,
		\begin{align}
			(M,R)=\left(\frac{2t}{K}+(1-\frac{t}{K})(\mu_1/2+\mu_2), (\mu_1 +\mu_2/2)\frac{K-t}{t+1}\right), \forall t\in [0:K-1], 
		\end{align} 
for $\mu_1 +\mu_2=1$ and 	$\mu_1 , \mu_2>0$, 
	while the required subpacketization is $\binom{K}{t}$; 
	
		2) when $N=3$, the following memory-load points are achievable,
		\begin{align}
			(M,R)=\left(\frac{3t}{K}+(1-\frac{t}{K}), \frac{K-t}{t+1}\right), \forall t\in [0:K-1], 
		\end{align} 
	while the required subpacketization is $\binom{K}{t}$;
	
		3) when $N=4$, the following memory-load points are achievable,
		\begin{align}
			(M,R)=\left(\frac{4t}{K}+(1-\frac{t}{K}), \frac{K-t}{t+1}\right), \forall t\in [0:K-1], 
		\end{align} 
while the required subpacketization is $\binom{K}{t}$;

		4) when general $N$, the following memory-load points are achievable,
		\begin{align}
			(M,R)=\left(\frac{t}{K}N+(1-\frac{t}{K})O(\sqrt{N}), O(\sqrt{N}) \frac{K-t}{t+1})\right), \forall t\in [0:K-1],
			\label{eq:general N}
		\end{align} 
		while the required subpacketization is $O\left(\binom{K}{t} \sqrt{N}^{\sqrt{N}}\right)$. 
	\end{corollary}

For the general $N$, the proposed caching scheme with private demands and caches in~\eqref{eq:general N} has subpacketization level   $O\left(\binom{K}{t} \sqrt{N}^{\sqrt{N}}\right)$. Note that the subpacketization of the virtual users scheme in Theorem~\ref{thm:virtual user} is $\binom{NK}{t}$. Based on the asymptotic approximation of the binomial coefficients, the subpacketization of the virtual users scheme would be $F_1=\binom{NK}{MK} \simeq 2^{NK\mathcal{H}(\frac{M}{N})}$, where $\mathcal{H}(.)$ is the binary entropy function. The subpacketization of our general scheme is on the order of $F_2 \simeq 2^{K\mathcal{H}(\frac{M}{N})} 2^{\frac{1}{2}\sqrt{N}\log_2(N)}$.
	 Then 
	\begin{align}
		\frac{F_2}{F_1} = \frac{2^{K\mathcal{H}(\frac{M}{N})+\frac{1}{2}\sqrt{N}\log_2(N)}}{2^{NK\mathcal{H}(\frac{M}{N})}}.
	\end{align}
	If we assume $\frac{M}{N}$ is not vanishing with $N$,  $\frac{F_2}{F_1}$ goes to $0$ when $N$ and $K$ increase.
	
	\begin{remark}
		Based on Remark \ref{re:simp}, the proposed demand and cache private coded caching schemes in Corollary \ref{cor:tradeoffdiffN} for $N \in \{2,3,4\}$, are optimal within a constant multiplicative factor.
	\end{remark}

At the end of this subsection, we   illustrate the main idea of the construction in Theorem~\ref{thm:private demands} through one example.
	\begin{example}[$K=N=2, M=\frac{5}{4}$]
  In this example, we use the PIR scheme in     \cite[Section III-A]{tian2019capacity} in which the total download cost is $R_D=3/2$ and the subpacketization level is $F'=1$. Their scheme is presented in Table \ref{tab:tablentian}.
	
	\begin{table}
		\centering
		\begin{tabular}{| c || c || c | c |} 
			\hline
			\multirow{2}{*}{} & \multirow{2}{*}{Server $1$} &  \multicolumn{2}{c|}{Server $2$} \\ \cline{3-4}
			& & $d=1$ & $d=2$ \\
			\hline\hline
			$T=0$ & $0$ & $W_1$ & $W_2$ \\
			\hline
			$T=1$ & $W_1+W_2$ & $W_2$ & $W_1$ \\
			\hline
		\end{tabular}
		\caption{ Two-server PIR scheme in \cite{tian2019capacity} for $N=K=2$.}
  \label{tab:tablentian}
	\end{table}

	Assume the two files are $A$ and $B$. Each file is devided into two equal-length and non-overlapping subfiles as $A=(A_1,A_2)$ and $B=(B_1,B_2)$. 
	
	{\it Placement.}
	For the first part of the cache, user 1 caches $Z_1=(A_1,B_1)$ and user 2 caches $Z_2=(A_2,B_2)$. As can be seen in the PIR scheme, $\mathcal{Q}_1=\mathcal{Q}_2=2$. User $i$ chooses $T_i\in \{0,1\}$ each with probability $1/2$. Suppose $T_1=0$ and $T_2=1$. Based on our proposed approach in Theorem~\ref{thm:cachepir}, the second user additionaly caches $\gamma_{1}(Q_{1,2}=T_2=1,A_1,B_1)=A_1+B_1$, while the first user caches nothing additional since $\gamma_{1}(Q_{1,1}=T_1=0,A_2,B_2)=0$. So in total, the caches by the two users are
	\begin{align}
		Z_1&=(A_1,B_1), \\
		Z_2&=(A_2,B_2,A_1+B_1).
	\end{align}
	
	{\it Delivery.}
	Assume that  user 1 demands file $A$ and user 2 demands file $B$. 
	 Since $\gamma_{2}(Q_{2,1}=T_1=0,A_2,B_2)=A_2$ and $\gamma_{2}(Q_{2,2}=T_2=1,A_1,B_1)=A_1$, the transmission of the server is $A_2+A_1$. User $1$ cancels out $A_1$ and recovers $A_2$.   User $2$ 
recovers $B_1$ by using  the transmission 	$A_2+A_1$ and the cached content $A_2$,  $A_1+B_1$. 
	 So both users receive their desired subfiles. For other cases of $(T_1,T_2)$, the transmission of the server follows Table \ref{tab:tableex}. As can be seen, when $A_2+A_1$ is sent by the server, there can be four different cases happening.
	
	\begin{itemize}
		\item $(T_1,T_2)=(0,0)$ and demand vector $\mathbf{d}=(A,A)$;
		\item $(T_1,T_2)=(0,1)$ and demand vector $\mathbf{d}=(A,B)$;
		\item $(T_1,T_2)=(1,0)$ and demand vector $\mathbf{d}=(B,A)$;
		\item $(T_1,T_2)=(1,1)$ and demand vector $\mathbf{d}=(B,B)$.
	\end{itemize}
	For user 1 who is aware of the values $T_1=0, d_1=A$, there can exist two possible options of  
	\begin{itemize}
		\item $(T_1,T_2)=(0,0)$ and demand vector $\mathbf{d}=(A,A)$,
		\item $(T_1,T_2)=(0,1)$ and demand vector $\mathbf{d}=(A,B)$,
	\end{itemize}
	which reveals no information about the value of $d_2$ nor $T_2$ since
	\begin{align}
		& \Pr (d_2=A|T_1=0, d_1=A, X_\mathbf{d}=A_2+A_1) = \frac{\Pr (d_2=A, T_1=0, d_1=A, X_\mathbf{d}=A_2+A_1)}{\Pr (T_1=0, d_1=A, X_\mathbf{d}=A_2+A_1)} \nonumber \\
		& = \frac{\Pr (d_2=A, T_1=0, d_1=A) \Pr (X_\mathbf{d}=A_2+A_1 | d_2=A, T_1=0, d_1=A)}{\Pr (T_1=0, d_1=A) \Pr (X_\mathbf{d}=A_2+A_1|T_1=0, d_1=A)} \nonumber \\
		& = \frac{(1/2)^3 (1/2)}{(1/2)^2(1/2)}=\frac{1}{2},
	\end{align}
	and
	\begin{align}
		& \Pr (T_2=0|T_1=0, d_1=A, X_\mathbf{d}=A_2+A_1) = \frac{\Pr (T_2=0, T_1=0, d_1=A, X_\mathbf{d}=A_2+A_1)}{\Pr (T_1=0, d_1=A, X_\mathbf{d}=A_2+A_1)} \nonumber \\
		& = \frac{\Pr (T_2=0, T_1=0, d_1=A) \Pr (X_\mathbf{d}=A_2+A_1 | T_2=0, T_1=0, d_1=A)}{\Pr (T_1=0, d_1=A) \Pr (X_\mathbf{d}=A_2+A_1|T_1=0, d_1=A)} \nonumber \\
		& = \frac{(1/2)^3 (1/2)}{(1/2)^2(1/2)}=\frac{1}{2},
	\end{align}
	which equlas the prior probability for $d_2$ and $T_2$. Thus, both the demand and cache of user 2 is kept private. Similarly this holds for user 1.
	
	\begin{table}
		\centering
		\begin{tabular}{| c || c | c | c | c |} 
			\hline
			& $\mathbf{d}=(A,A)$ & $\mathbf{d}=(A,B)$ &  $\mathbf{d}=(B,A)$ & $\mathbf{d}=(B,B)$ \\ 
			\hline\hline
			$(T_1,T_2)=(0,0)$ & $A_2+A_1$ & $A_2+B_1$ & $B_2+A_1$ & $B_2+B_1$ \\
			\hline
			$(T_1,T_2)=(0,1)$ & $A_2+B_1$ & $A_2+A_1$ & $B_2+B_1$ & $B_2+A_1$ \\
			\hline
			$(T_1,T_2)=(1,0)$ & $B_2+A_1$ & $B_2+B_1$ & $A_2+A_1$ & $A_2+B_1$ \\
			\hline
			$(T_1,T_2)=(1,1)$ & $B_2+B_1$ & $B_2+A_1$ & $A_2+B_1$ & $A_2+A_1$ \\
			\hline
		\end{tabular}
		\caption{Delivery phase of demand and cache private coded caching scheme for $K=N=2$ and $M=\frac{5}{4}$.}
  \label{tab:tableex} 
	\end{table}

	Note that both the load $1/2$ and cache size $5/4$ are expected values over the random choice of the queries to the first server in the placement phase and the corresponding queries to the second server in the delivery phase. Note that user $2$ in this example has a cache size of $3/2$ but if it had chosen $T_2=0$, like the first user, it would have had a cache of size $1$. So on average we have a cache size of $5/4$.
	
	As a comparison, the privacy key scheme in~\cite{yan2021fundamental} for the same system parameters of $K=N=2, M=5/4$ has a load of $R=5/4$ whereas our scheme reaches the load $R=1/2$ while it preserves cache privacy additionally and the privacy key scheme does not. 
	
	\end{example}



	\section{Coded Caching with Private Demands and Imperfectly Private Caches} \label{sec:extension}
	Since constructing two-server PIR schemes under the UDIQ property is difficult, and in any case the download cost $R_D$ increases at least as $O(\sqrt{N})$ (see Theorem~\ref{thm:lowerbound}), to be able to propose better PIR schemes in terms of download cost which results in better memory-load tradeoffs for the corresponding caching scheme (see Theorem~\ref{thm:private demands}), in this section we relax the perfect cache privacy and allow some leakage in the cache information, while preserving perfect demand privacy.  

	We first review the leakage metric in the literature of  leaky PIR and then introduce our metric of leakage.  
	  Next, we apply the two-server PIR scheme in~\cite{dvir20162} into our construction structure in Theorem~\ref{thm:private demands}, and compute the cache leakage of the resulting coded caching scheme with private demands. Finally, we compare the resulting schemes with the existing coded caching schemes with private demands, in terms of load and cache leakage.

	\subsection{Cache information leakage}
	Privacy leakage has been already introduced in several works on PIR following various definitions (see~\cite{samy2019capacity,dwork2008differential,lin2019weakly,lin2021capacity,samy2021asymmetric, guo2020information, zhou2020weakly}). In this section we introduce a privacy leakage definition on the cache information which is relevant to our setting. 
    The decoding and demand privacy constraints stay the same as in \eqref{eq:decoding} and \eqref{eq:demand privacy constraint}, while the cache privacy constraint in \eqref{eq:cache privacy constraint} does not exist anymore. As the cache privacy constraint in \eqref{eq:cache privacy constraint} suggests, the perfect scenario for cache is that the ambiguity on its information does not change conditioned on the knowledge of server transmission. In a non-perfect scenario, we want to keep the distribution on cache information before and after server transmission close to each other as much as possible. 
  
    In information-theoretic secrecy \cite{el2011network}, the {\it information leakage rate} associated with the $(2^{nR}, n)$ secrecy code is defined as
    \begin{align}
   	  \frac{1}{n} I(M,Z^n),
    \end{align}
	in which $M$ represents the sender's message and $Z^n$ represents the message received by the eavesdropper for the block length $n$. In our definition for cache privacy, the server's transmission $X_{\mathbf{d}}$ acts as the message received by the eavesdropper, and user $k$'s cache metadata $\mathscr{M}_k$ the message we want to keep private. We replace the block length $n$ with the entropy of cache metadata as the block length for the message. This motivates our consideration of the following cache leakage metric for user $k$:
	\begin{align} \label{def:eps}
		\epsilon_k =  \frac{I(\mathscr{M}_k; X_{\mathbf{d}})}{H(\mathscr{M}_k)} = 1 - \frac{H(\mathscr{M}_k |  X_{\mathbf{d}})}{H(\mathscr{M}_k)}, \ \forall k \in [K]
	\end{align}
	where $H(\cdot)$ is the entropy function. In the fully private case when there is no leakage, this metric is $0$. As the uncertainty amount on cache information decreses after server transmission, the leakage grows and goes to $1$ when the cache information is fully leaked.

\subsection{Cache-leakages of \cite{yan2021fundamental} and \cite{DBLP:journals/corr/abs-1909-03324}}
	
	We then consider the coded caching schemes with private demands in \cite{yan2021fundamental} and \cite{DBLP:journals/corr/abs-1909-03324}, and compute their leakages on the cache. 
		 For the case of single-file requests in \cite{yan2021fundamental}, the randomness on cache for user $k$ is $\mathscr{M}_k=\mathbf{p}_k:=(p_{k,1}, \dots, p_{k,N})$ which is chosen uniformly at random from $\mathbb{F}_q^k$,  such that the summation of the elements of $\mathbf{p}_k$ equals $q-1$; $\sum_{n \in [N]} p_{k,n}=q-1$. Based on this constraint, the total number of choices for $\mathbf{p}_k$ is $q^{N-1}$. Thus we have
	\begin{align}
		H(\mathscr{M}_k) = (N-1) \log (q).
	\end{align}
	If we denote the demand vector for user $k$ by $\mathbf{d}_k$ which for single-file demands has a $1$ on the position of requested file index and $0$ elsewhere, the server sends $\mathbf{q}_k=\mathbf{p}_k+\mathbf{d}_k$ for all $k \in [K]$ as a metadata alongside the main message. Having $\mathbf{q}_k$, since there are only $N$ options for $\mathbf{d}_k$ (uniformly chosen), our options for $\mathbf{p}_k$ would be also limited to $N$. Thus
	\begin{align}
		H(\mathscr{M}_k | X_{\mathbf{d}}) = \log (N).
	\end{align}
	According to \eqref{def:eps} we have
	\begin{align} \label{eq:leakagepk}
		\epsilon_k = 1- \frac{1}{\log (q)} \frac{\log (N)}{(N-1)},
	\end{align}
	which goes to $1$ as $N$ increases. Our goal is to introduce a coded caching scheme with non-zero leakage on cache using a two-server PIR scheme that does not satisfy the UDIQ condition in Definition \ref{def:udiq}, instead of the perfectly private scheme of Theorem \ref{thm:cachinginPIR}, with the benefit of attaining better download costs and subpacketization for the PIR scheme which will directly affect the memory-load tradeoff and subpacketization of the resulted coded caching scheme based on our structure in Theorem \ref{thm:cachepir}. 
	
	For the virtual users scheme of \cite{DBLP:journals/corr/abs-1909-03324}, the cache of user $k$ is selected between $N$ choices uniformly at random. After the transmission, the probabilty distribution over cache information does not change as proved in Theorem \ref{thm:virtual user}. So in this case, the leakage would be $\epsilon_k=0$ for all users, which is perfect but as mentioned before, this scheme has a huge subpacketization level. 	

	\subsection{Review on \cite{woodruff2005geometric} and \cite{dvir20162}} \label{sec:bestpir}
	We then review the protocol proposed in \cite{dvir20162} with the lowest communication cost (equal to $N^{o(1)}$) among all existing two-server PIR protocols, which will be applied into our proposed construction structure in Theorem~\ref{thm:private demands}.
	This scheme is a combination of an existing two-server PIR scheme which uses polynomial interpolation \cite{woodruff2005geometric} and Matching Vector Codes (MV codes) \cite{efremenko20123, yekhanin2008towards}. We will shortly go through \cite{woodruff2005geometric} and then introduce matching vector families and after that, describe the protocol in \cite{dvir20162}. 
	
	The scheme in \cite{woodruff2005geometric} is based on building  polynomials with degree $3$. First, choose $k$ such that $N \leq \binom{k}{3}$. Pick a finite field $\mathbb{F}_q$ where $q>3$. Define an encoding $\phi$ that maps indices in $[N]$ to binary $k$-dimensional space. 
	\begin{align}
		\phi:[N] \rightarrow \{0,1\}^k \subset \mathbb{F}_q^k,
	\end{align}
	such that the resulting $k$-dimensional codewords are of Hamming weight $3$. If we denote the $k$-dimantional space by $\mathbf{x}=(x_1,\dots,x_k)$, the polynomial $F(\mathbf{x}) \in \mathbb{F}_q[x_1,\dots,x_k]$ where $\mathbb{F}_q[x_1,\dots,x_k]$ denotes the field of polynomials with variables $x_1,\dots,x_k$ over $\mathbb{F}_q$, is defined as follows, 
	\begin{align}
		F(\mathbf{x}) = \sum_{i=1}^{N} W_i \left(\prod_{j:\phi(i)_j=1}x_j\right),
	\end{align}
	in which the files $W_i$ are considered to be one bit. This polynomial satisfies $F(\phi(i))=W_i, \forall i \in [N]$.
	
	Suppose the user demands  the file $W_\tau$. The scheme works as follows:
	\begin{itemize}
		\item the user picks a $\mathbf{z} \in \mathbb{F}_q^k$ uniformly at random;
		\item the user sends $\phi(\tau) + t_i\mathbf{z}$ to server $i$ in which $t_1 \neq t_2 \in \mathbb{F}_q \backslash \{0\}$;
		\item server $i$ sends to the user the values $F(\phi(\tau)+t_i\mathbf{z})$ and $\nabla F(\phi(\tau)+t_i\mathbf{z})$.
	\end{itemize} 
	With the answers received from both servers,   the user can retrieve $F(\phi(\tau))=W_\tau$; the reader can refer to \cite{woodruff2005geometric} for the detailed proof of decodability. 
	 The privacy of demand is protected since $\phi(\tau) + t_i\mathbf{z}$ is uniformly distributed in $\mathbb{F}_q^k$ for any value of $\tau$. 
	
	We then review the two-server PIR scheme in \cite{dvir20162}, starting with the following definition. 
	\begin{definition} [Matching Vector Family] \label{def:mvf}
		Let $S \subset \mathbb{Z}_m \backslash \{0\}$ and let $\mathcal{F} = (\mathcal{U}, \mathcal{V})$ where $\mathcal{U}=(\mathbf{u}_1,\dots,\mathbf{u}_N)$, $\mathcal{V}=(\mathbf{v}_1,\dots,\mathbf{v}_N)$ and $\mathbf{u}_i, \mathbf{v}_i \in \mathbb{Z}_m^w, \forall i \in [N]$. Then $\mathcal{F}$ is called an $S$-matching vector family of size $N$ and dimension $w$ if $\forall i,j$,
		\begin{align}
			\langle \mathbf{u}_i,\mathbf{v}_j \rangle \left\{
			\begin{array}{ll}
				=0 \; &\text{if } i=j \\
				\in S &\text{if } i\neq j
			\end{array},
			\right.
		\end{align}
	\end{definition}
	where $\langle \mathbf{u}_i,\mathbf{v}_j \rangle$ indicates the inner product between the two vectors. 
	It has been shown that based on \cite[Theorem 1.2]{grolmusz2000superpolynomial}, for $S=\{1,3,4\}$, we can build matching vector codes with parameters $N$ and $w$ (and when $m$ is composite) such that 
	\begin{align} \label{w}
		w=\exp \left(O\left(\sqrt{\log N \log \log N}\right)\right).
	\end{align}
	For a commutative ring $\mathcal{R}$, the ring of polynomials in variables $x_1, \dots, x_w$ with coefficients in $\mathcal{R}$ is denoted by $\mathcal{R}[x_1,\dots,x_w]$. In \cite{dvir20162}, the authors introduce a definition to extend the notion of partial derivatives to polynomials in $\mathcal{R}[x_1,\dots,x_w]$ as follows.
	\begin{definition}
		Let $\mathcal{R}$ be a commutative ring and let $F(\mathbf{x})=\sum c_{\mathbf{z}} \mathbf{x}^{\mathbf{z}} \in \mathcal{R}[x_1,\dots,x_w]$. We define $F^{(1)}(\mathbf{x}) \in (\mathcal{R}^w)[x_1,\dots,x_w]$ to be
		\begin{align} \label{eq:polynomialder}
			F^{(1)}(\mathbf{x}):=\sum (c_{\mathbf{z}} \cdot \mathbf{z}) \mathbf{x}^{\mathbf{z}},
		\end{align}
	\end{definition}
	where $ \mathbf{x}^{\mathbf{z}}=x_1^{z_1} x_2^{z_2} \dots x_w^{z_w}$. Now we are ready to introduce the scheme in \cite{dvir20162}.
    
    For the rest of this section, $\mathcal{R}=\mathcal{R}_{6,6}=\mathbb{Z}_6[\gamma]/(\gamma^6-1)$ which is the ring of univariate polynomilas $\mathbb{Z}_6[\gamma]$ modulo the identity $\gamma^6=1$ as defined in \cite{dvir20162}. It should be pointed out that the set $S$ which contains only three values is the key to this scheme since, rougly speaking, the powers of $\gamma$ appearing in the polynomial are from this set and $0$ and therefore, there will be four unknown coefficients and we would only need two evaluations and two derivatives to recover the intended value. We will not go into the details of the recovery and refer the reader to the paper. 
    
    Assume the user's demand is $W_\tau$. The servers save the data in the polynomial $F(\mathbf{x}) \in \mathcal{R}[x_1,\dots,x_w]$ where
	\begin{align} \label{eq:polynomial}
		F(\mathbf{x})=F(x_1,\dots,x_w)=\sum_{i=1}^{N}W_i \mathbf{x}^{\mathbf{u}_i},
	\end{align}
	in which $\mathcal{U}=(\mathbf{u}_1,\dots,\mathbf{u}_N)$ is given by the matching vector family $\mathcal{F} = (\mathcal{U}, \mathcal{V})$ for $m=6$ as $w$ as in~\eqref{w}. Then,
	\begin{itemize}
		\item the user picks a $\mathbf{z} \in \mathbb{Z}_6^w$ uniformly at random;
		\item the user sends $\mathbf{z}+t_i \mathbf{v}_\tau$ to server $i$; 
		\item server $i$ sends back the values $F(\gamma^{\mathbf{z}+t_i \mathbf{v}_\tau})$ and $F^{(1)}(\gamma^{\mathbf{z}+t_i \mathbf{v}_\tau})$,
	\end{itemize}
	where the vector $(\gamma^{z_1+t_i \mathbf{v}_{\tau,1}}, \gamma^{z_2+t_i \mathbf{v}_{\tau,2}}, \dots, \gamma^{z_w+t_i \mathbf{v}_{\tau,w}})$ is denoted by $\gamma^{\mathbf{z}+t_i \mathbf{v}_\tau}$. Since the values $\mathbf{z}+t_i \mathbf{v}_\tau$ are distributed uniformly on $\mathbb{Z}_6^w$, the privacy of demand in the PIR scheme is preserved. Also in the scheme $t_1=0$ and $t_2=1$. Since the user sends elements in $\mathbb{Z}_6^w$ to both servers and recieves an element in $\mathcal{R}$ and another one in $\mathcal{R}^w$ from each server, the communication cost would be $O(w)=N^{O\left(\sqrt{\frac{\log \log N}{\log N}}\right)}$.

	\subsection{Coded caching schemes with private demands and imperfectly private caches based on Theorem~\ref{thm:private demands}}

	We now apply the two-server PIR scheme  in~\cite{dvir20162} into our structure in Theorem \ref{thm:private demands}. 	Assume we have a system of $K$ users and $N$ files $W_1,\dots,W_N$. The server is connected to the users with a shared link. For any $t =KM/N \in [K]$, each file is split into $\binom{K}{t}$ non-overlapping subfiles with the same size, 
	\begin{align}
		W_n=(W_{n,\tau}:\tau \subset [K], |\tau|=t).
	\end{align}
	We assume that each   subfile  has one bit; but we can easily extend the scheme for the other case. 
	
{\it Placement phase.} 	In the first part of the placement phase, for every $k \in [K]$, any subfile $W_{n,\tau}$ is stored in the cache if $k \in \tau$. Therefore,
	\begin{align}
		\{W_{n,\tau}: n\in [N], \tau \subset [K], |\tau|=t, k\in\tau  \} \subset Z_k.
	\end{align}
	
	In the second part of the placement phase, for each set $\tau \subseteq [K]$ where $|\tau|=t$ and  $k \notin \tau$, user $k$ caches the result of an encoding on all subfiles $\{W_{n,\tau}, n \in [N] \}$. The matching vector family $\mathcal{F}=(\mathcal{U}, \mathcal{V})$ is constructed in which $\mathcal{U}=(\mathbf{u}_1,\dots,\mathbf{u}_N)$ and $\mathcal{V}=(\mathbf{v}_1,\dots,\mathbf{v}_N)$ such that $\mathbf{u}_i, \mathbf{v}_i \in \mathbb{Z}_6^w, \forall i \in [N]$ as explained previously.
	
	User $k$ picks $\mathbf{z}_k \in \mathbb{Z}_6^w$ uniformly at random. The user sends $\mathbf{z}_k$ to the server. For each $\tau$ such that $k \notin \tau$, the server sends  $F(\gamma^{\mathbf{z}_k},W_{[N],\tau})$ and $F^{(1)}(\gamma^{\mathbf{z}_k},W_{[N],\tau})$ to the user where 
	\begin{align} \label{eq:defencoding}
		F(\mathbf{x},W_{[N],\tau})&=F(x_1,\dots,x_k,W_{1,\tau},\dots,W_{N,\tau})\nonumber \\&=\sum_{i=1}^{N}W_{i,\tau} \mathbf{x}^{\mathbf{u}_i}.
	\end{align}
	This completes the placement phase. 
	
	{\it Delivery phase.} 	Assume that user $k$ demands the file $W_{\tau_k}$. In the delivery phase, user $k$ sends $\mathbf{z}_k+\mathbf{v}_{\tau_k}$ to the server. For each $\mathcal{S} \subset [K]$ where $|\mathcal{S}|=t+1$, the server sends the multicast messages
	\begin{align}
		 Y_{\mathcal{S}} =  
		&\left(\sum_{s \in \mathcal{S}} F\left(\gamma^{\mathbf{z}_s+\mathbf{v}_{\tau_s}},W_{[N],\mathcal{S}\backslash s}\right), \sum_{s \in \mathcal{S}} F^{(1)}\left(\gamma^{\mathbf{z}_s+\mathbf{v}_{\tau_s}},W_{[N],\mathcal{S}\backslash s}\right) \right).
	\end{align}

	Alongside with messages $Y_{\mathcal{S}}$, in order for the users to be able to decode their needed messages, the server should send the values $\{\mathbf{z}_k+\mathbf{v}_{\tau_k}, \forall k \in [K]\}$ as metadata. So the transmitted message by the server $X_{\mathbf{d}}$ would be
	\begin{align}
		X_{\mathbf{d}} &= \{Y_{\mathcal{S}}, \mathcal{S} \subseteq [K], |\mathcal{S}|=t+1\}   \bigcup \{\mathbf{z}_k+\mathbf{v}_{\tau_k}, \forall k \in [K]\}.
	\end{align}
	The decodability proof follows from the proof of Theorem \ref{thm:cachepir}.
	
	{\it Performance.} 
	An observation on the scheme reveals that the cache metadata equals $\mathscr{M}_k=\mathbf{z}_k$. Now we compute the amount of cache information leakage based on \eqref{def:eps}. Since $\mathscr{M}_k$ is uniformly distributed in $\mathbb{Z}_6^w$,
	\begin{align}
		H(\mathscr{M}_k)=w \log 6.
	\end{align}
	In addition, we have
	\begin{subequations}
		\begin{align}
			H(\mathscr{M}_k | X_{\mathbf{d}})&  = H\left(\mathbf{z}_k | \{\mathbf{z}_{k'}+\mathbf{v}_{\tau_{k'}}, \forall k' \in [K]\}\right) \label{p1} \\
			&= H\left(\mathbf{z}_k | \mathbf{z}_k+\mathbf{v}_{\tau_k} \right) \label{p2} \\
			&= \log N,
		\end{align}
	\end{subequations}
	where \eqref{p1} comes from that the values $Y_{\mathcal{S}}$ depend on $\{\mathbf{z}_{k'}+\mathbf{v}_{\tau_{k'}}, \forall k' \in [K]\}$ and the library and \eqref{p2} comes from independence of $\mathbf{z}_k$ values for all $k \in [K]$. Therefore
	\begin{align} \label{leakours}
		\epsilon_k = 1 - \frac{\log N}{w \log 6} = 1 - O\left(\frac{\log N}{ N^{\sqrt{\frac{\log \log N}{\log N}}}}\right). 
	\end{align}
		Since the total download cost of this scheme in $O(w)=N^{O\left(\sqrt{\frac{\log \log N}{\log N}}\right)}$, based on our structure in Theorem \ref{thm:cachepir}, we can achieve the lower convex envelop of the the memory-load pair points 
		\begin{align}
			(M,R)= \left(  \frac{Nt}{K}+\left(1-\frac{t}{K}\right) N^{O\left(\sqrt{\frac{\log \log N}{\log N}}\right)}, N^{O\left(\sqrt{\frac{\log \log N}{\log N}}\right)}\frac{K-t}{t+1}  \right).
		\end{align}

	To compare, the memory-load tradeoff and cache leakage of the privacy key scheme of \cite{yan2021fundamental} follows $\left(1+\frac{t(N-1)}{K}, \frac{\binom{K}{t+1}-\binom{K-\min \{N-1,K\}}{t+1}}{\binom{K}{t}}\right), \ \forall t\in [0:K]$ and $\epsilon_k = 1- \frac{1}{ \log (q)} \frac{\log (N)}{(N-1)}$ respectively, whereas for our scheme for general $N$ in Corollary \ref{cor:tradeoffdiffN}, we have $\left(\frac{t}{K}N+(1-\frac{t}{K})O(\sqrt{N}), O(\sqrt{N}) \frac{K-t}{t+1}\right), \\ \forall t\in [0:K-1]$ as memory-load pair and $\epsilon_k=0$. For the scheme in this section these parameters are $\left(  \frac{Nt}{K}+\left(1-\frac{t}{K}\right) N^{O\left(\sqrt{\frac{\log \log N}{\log N}}\right)}, N^{O\left(\sqrt{\frac{\log \log N}{\log N}}\right)}\frac{K-t}{t+1}  \right), \forall t\in [0:K-1]$ and $\epsilon_k = 1 - O\left(\frac{\log N}{ N^{\sqrt{\frac{\log \log N}{\log N}}}}\right)$. The scheme in this section works better in terms of cache leakage than the privacy key scheme since it converges much more slowly to $1$, but has worse load. On the other hand, compared to our perfectly private scheme, it has a better load but of course worse leakage on cache. 
	
	At the end of this section, we provide an example to illustrate the proposed coded caching scheme with private demands and imperfectly private caches by leveraging the two-server PIR scheme  in~\cite{dvir20162}. 
	\begin{example}[$N=K=2$ and $t=1$]
This is an example just to demonstrate the placement and delivery phases of the proposed scheme. Therefore, we will not care about the $S$ and $w$ parameters of the matching vector family. In this scheme, $S=\{1\}$ and $w=2$.
We consider the coded caching problem with $N=K=2$ and $t=KM/N=1$. 
Each file is splitted into $\binom{K}{t}=2$ non-overlapping equally-sized subfiles, i.e. $A=(A_1, A_2), B=(B_1,B_2)$. In the first part of the placement phase, each user's cache will be as follows, 
	\begin{align}
		Z_1&=(A_1,B_1), \\
		Z_2&=(A_2,B_2).
	\end{align}
	
	For the second part of placement, we first introduce a matching vector family based on Definition \ref{def:mvf}. 
	We define the $2$-tuples $\mathcal{U}$ and $\mathcal{V}$ of the matching vector family as follows,
	\begin{align}
		\mathcal{U} &= \left((0,1),(1,0)\right), \\
		\mathcal{V} &= \left((1,0),(0,1)\right).
	\end{align}

	The polynomial $F(\mathbf{x})$ in \eqref{eq:polynomial} for library files $W_1$ and $W_2$ is as follows,
	\begin{align}
		F(\mathbf{x}) &= W_1x_2+W_2x_1
	\end{align}

	In addition, the function $F^{(1)}(\mathbf{x})$ in \eqref{eq:polynomialder} would be, 
	\begin{align}
		F^{(1)}(\mathbf{x}) = (W_2x_1,W_1x_2)
	\end{align}

	For the second part of the placement phase, user $k$ chooses $\mathbf{z}_k \in \mathbb{Z}_6^2$ uniformly at random. Suppose the choices are $\mathbf{z}_1=(2,3), \mathbf{z}_2=(5,1)$. Users send these values to the server. The server sends back the pair $\left(F(\gamma^{\mathbf{z}_1}), F^{(1)}(\gamma^{\mathbf{z}_1})\right)$ to user 1 when $W_1=A_2, W_2=B_2$ and $\left(F(\gamma^{\mathbf{z}_2}), F^{(1)}(\gamma^{\mathbf{z}_2})\right)$ to user 2 when $W_1=A_1, W_2=B_1$. So in total the caches are as follows, 
	\begin{align}
		Z_1&=(A_1,B_1, A_2\gamma^3+B_2\gamma^2, (B_2\gamma^2,A_2\gamma^3)), \\
		Z_2&=(A_2,B_2, A_1\gamma+B_1\gamma^5, (B_1\gamma^5,A_1\gamma)).
	\end{align}

	In the delivery phase, suppose the demand for users 1 and 2 are $A,B$, respectively. When the server receives the demands, it should compute for user $1$ the values $\left(F(\gamma^{\mathbf{z}_1+\mathbf{v}_1}), F^{(1)}(\gamma^{\mathbf{z}_1+\mathbf{v}_1})\right)$ when $W_1=A_2, W_2=B_2$ and for user $2$ the values $\left(F(\gamma^{\mathbf{z}_2+\mathbf{v}_2}), F^{(1)}(\gamma^{\mathbf{z}_2+\mathbf{v}_2})\right)$ when $W_1=A_1, W_2=B_1$. Then adds each part together and sends the multicast messages
	\begin{align}
		(A_2\gamma^3+B_2\gamma^3)+(A_1\gamma^2+B_1\gamma^5), \\
		\left( B_2\gamma^3+B_1\gamma^5, A_2\gamma^3+A_1\gamma^2 \right),
	\end{align}
	including the metadata to the users on the shared channel. Using this transmission and its cache content, user 1 recovers $A_2\gamma^3+B_2\gamma^3$ and $(B_2\gamma^3, A_2\gamma^3)$ and user 2 recovers $A_1\gamma^2+B_1\gamma^5$ and $(B_1\gamma^5, A_1\gamma^2)$ and using the decoding procedure for the PIR scheme, each user can recover its demanded file.  
	 The privacy of demands are fully satisfied; this is because,   from the metadata $\mathbf{z}_2+\mathbf{v}_2=(5,2)$, user 1 would not know any information about the value $\mathbf{v}_2$ since $\mathbf{z}_2$ is uniformly distributed on $\mathbb{Z}_6^2$. On the other hand, the cache is not perfectly private. The cache leakage in this example equals $\epsilon_k = 1- \frac{\log N}{w \log 6}=1 - \frac{\log 2}{2 \log 6}$.
	\hfill $\square$  
\end{example}	
	
	\section{Conclusion} \label{sec:conclusion}
	In this paper, we formulated the  coded caching problem with private demands and caches, where we added one privacy constraint on users' caches into the existing coded caching problem with private demands. 
	We first showed that the existing demand-private coded caching based on introducing virtual users can also preserve the privacy of caches, while suffering from a super high subpacketization. The main contribution of this paper was to propose a new  structure on constructing coded caching schemes with private demands and caches by using two-server PIR schemes with uniform demands and independent queries. 
We provided the construction of new two-server PIR schemes with this condition and by applying them into the coded caching scheme construction, we have obtained new schemes with significant 	
 reduction on the subpacketization compared to the virtual users scheme. We have also provided a lower bound on the download cost of these PIR schemes and a matching achievable scheme. We then extended the proposed structure to the coded caching problem with private demands and imperfectly private caches. Future and on-going works include providing a lower bound on the memory-load tradeoff of demand and cache private caching schemes, designing two-server PIR schemes for general $N$ with less subvpacketization compared to the proposed one, studying the tradeoff between the amount of leakage and system parameters of the PIR scheme in the imperfect private caches scenario. 
	
\appendices	
	
	
	\section{Proof of Theorem \ref{thm:private demands}: New Construction on Coded Caching Schemes with Demand Privacy} \label{proofthm4}
	We assume that the set of queries sent to servers $1$ and $2$ in the PIR scheme are chosen from the sets $\mathcal{Q}_{1}$ and $\mathcal{Q}_{2}$ respectively. Note that if a PIR scheme with download costs $R_{D_1}$ and $R_{D_2}$ corresponding to servers $1$ and $2$ is achievable, then by a time sharing argument, the download cost pair $(R'_{D_1}, R'_{D_2}) = (\mu_1 R_{D_1}+\mu_2 R_{D_2}, \mu_1 R_{D_2}+\mu_2 R_{D_1})$ where $\mu_1,\mu_2 \in [0,1], \mu_1+\mu_2+1$ is also achievable.
	
	{\it Placement.}
	We split each file in two parts. The first part follows exactly the same process as the MAN placement phase. For each $t \in [0:K-1]$, each file is splitted into $\binom{K}{t}$ nonoverlapping subfiles with the same size, 
	\begin{align}
		W_n=(W_{n,\tau}:\tau \subset [K], |\tau|=t),
	\end{align} 
	where each subfile contains $F/\binom{K}{t}$ bits. 
	In addition,  for each $n \in [N]$ and each $\tau \subset [K]$ where $ |\tau|=t$, we divide  each $W_{n,\tau}$   into $F'$ non-overlapping subfiles $W_{n,\tau,\omega}$ of the same size,
	\begin{align} \label{def:secsplit}
		W_{n,\tau} = \{W_{n,\tau,\omega}, \omega \in [F']\},
	\end{align}
	where we recall that $F'$ represents the subpacketization of the two-server PIR scheme.  		
	
	Each user  $k \in [K]$ first caches    $W_{n,\tau}$ where $k \in \tau$; in other words, 
	\begin{align}
		\{W_{n,\tau}: n\in [N], \tau \subset [K], |\tau|=t, k\in\tau  \} \subset Z_k.
	\end{align}
	
	In addition, 
	for every index $\tau$ in which $k \notin \tau$, 
	user $k$ caches an encoding function on all subfiles $\{W_{n,\tau}, n \in [N] \}$. The encoding is chosen as follows. First, user $k$ chooses one query $Q_{1,k}$ from $\mathcal{Q}_{1}$ uniformly at random. Then the encoding fuction would be the answer of the first server in the PIR scheme when the query is $Q_{1,k}$ and the files are $W_{1,\tau},\ldots,W_{N,\tau}$, i.e. $\gamma_{1}(Q_{1,k}, W_{1,\tau},\ldots,W_{N,\tau})$. The second part of file splitting in \eqref{def:secsplit} is necessary to compute this encoding function. Thus, the second part of the cache for user $k$ would be
	\begin{align}
		\{ \gamma_{1}(Q_{1,k}, W_{1,\tau},\ldots,W_{N,\tau}), \tau \subset [K], |\tau|=t, k \notin \tau \} \subset Z_k.
	\end{align}
	Therefore in total, for every $k \in [K]$, $Z_k$ would be
	\begin{align} \label{eq:cache}
		Z_k = &\{W_{n,\tau}: n\in [N], \tau \subset [K], |\tau|=t, k\in\tau  \} \nonumber \\
		\bigcup & \{ \gamma_{1}(Q_{1,k}, W_{1,\tau},\ldots,W_{N,\tau}), \tau \subset [K], |\tau|=t, k \notin \tau \}.
	\end{align}
	totally containing $\frac{N \binom{K-1}{t-1}}{\binom{K}{t}} F + R'_{D_1} \frac{\binom{K-1}{t}}{\binom{K}{t} } F= \left( \frac{N t}{K} + R'_{D_1} \frac{K-t}{K} \right) F =MF$, satisfying the memory size constraint. 
	
	{\it Delivery.}
	Recall that for a $(N,K)$ MAN coded caching scheme, for each $\mathcal{S} \subset [K]$ such that $|\mathcal{S}|=t+1$, the server transmitts 
	\begin{align}
		\oplus_{s\in\mathcal{S}} W_{d_s,\mathcal{S}\backslash\{s\}},
	\end{align}
	where $\oplus$ stands for bitwise XOR. Instead in the delivery phase of our scheme, for each subset $\mathcal{S} \subseteq [K]$ where $|\mathcal{S}|=t+1$, the server transmits a multicast message as
	\begin{align}
		Y_{\mathcal{S}}=\Sigma_{s\in\mathcal{S}} \gamma_{2}\big(Q_{2,s}, W_{1,\mathcal{S}\backslash\{s\}},\ldots,W_{N,\mathcal{S}\backslash\{s\}}\big).
	\end{align}
	where $\gamma_2(.)$ is the answer encoding function of server 2 of the PIR scheme and $Q_{2,s}$ is chosen is such a way that the query pair $(Q_{1,s}, Q_{2,s})$ where $Q_{1,s}$ was chosen in the placement phase, corresponds to the $d_s^{th}$ message in the PIR problem. This means that for every $\tau$ such that $k \notin \tau$, the answer of server 1, $\gamma_{1}\big(Q_{1,k}, W_{1,\tau},\ldots,W_{N,\tau}\big)$, saved in the cache and the answer of server 2, $\gamma_{2}\big(Q_{2,k}, W_{1,\tau},\ldots,W_{N,\tau}\big)$, extracted from the message $Y_{\mathcal{S}}$ with $\mathcal{S}=\tau \cup \{k\}$, lead user $k$ to decode subfile $W_{d_k,\tau}$. Following the same process for all needed subfiles, user $k$ decodes file $W_{d_k}$. This proves the satisfaction of the decodability condition in \eqref{eq:decoding}. It can seen that in the delivery phase the server in total transmits $R'_{D_2} \frac{ \binom{K}{t+1}}{\binom{K}{t} } F =  R'_{D_2}\frac{K-t}{t+1} F$,  coinciding with~\eqref{eq:novel approach achieved memory load}.
	
	We should note that alongside the multicast messages, the server should send also the values $\{Q_{2,k}, k\in [K]\}$ as metadata so that everybody can decode their required messages. We assume that the size of this metadata is negligible compared to the multicast messages. Thus, the transmitted message $X_{\mathbf{d}}$ will be as follows, 
	\begin{align}
		X_{\mathbf{d}} &= \{Y_{\mathcal{S}}, \mathcal{S} \subseteq [K], |\mathcal{S}|=t+1\}   \bigcup \{Q_{2,k}, k \in [K]\}
	\end{align}
	Now we can check the demand privacy condition in \eqref{eq:demand privacy constraint}.
	\begin{subequations}
		\begin{align}
			&I(\mathbf{d}; X_{\mathbf{d}} | d_k, Z_k ) \nonumber\\ 
			&\leq I(\mathbf{d}; X_{\mathbf{d}} | d_k, Z_k, W_{[N]} ) \label{eq:demand privacy 1} \\
			&\leq I(\mathbf{d};  \{Q_{2,k}, k \in [K]\} | d_k, Z_k, W_{[N]}) \label{eq:demand privacy 2}\\
			&= \sum_{k'\in [K]\setminus \{k\}} I(d_{k'};  Q_{2,k'}  |   W_{[N]}) = 0   \label{eq:demand privacy 3}
		\end{align}
	\end{subequations}
	where~\eqref{eq:demand privacy 1} comes from \eqref{eq:independent},~\eqref{eq:demand privacy 2} comes from the fact that the set $\{Y_{\mathcal{S}}, \mathcal{S} \subseteq [K], |\mathcal{S}|=t+1\}$ is a function of  $ \{Q_{2,k}, k \in [K]\}$ and $W_{[N]}$,~\eqref{eq:demand privacy 3} comes from the fact that the pairs $(d_{i}, Q_{2,i})$ where $i\in [K]$ are independent of each other given $ W_{[N]}$ by our construction and the privacy constraint in \eqref{eq:PIR privacy constraint}.
	
	\section{Proof of Theorem \ref{thm:cachepir}} \label{proofthm5}
	
	Based on the proof  for Theorem \ref{thm:private demands}, we proved that our construction satisfies the decodability and demand privacy conditions in~\eqref{eq:decoding} and~\eqref{eq:demand privacy constraint} resepctively for any two-server PIR scheme. In this section, for PIR schemes that satisfy the UDIQ condition in Definition \ref{def:udiq} additionally, we just need to prove that the privacy condition in~\eqref{eq:cache privacy constraint} holds. 
	For the cache privacy constraint in~\eqref{eq:cache privacy constraint},  for each $k\in [K]$ we have
	\begin{subequations}
		\begin{align}
			&I\big( (\mathscr{M}_1,\ldots,\mathscr{M}_{K}); X_{\mathbf{d}} | d_k, Z_k\big) \nonumber \\
			&\leq  I\big( (\mathscr{M}_1,\ldots,\mathscr{M}_{K}); X_{\mathbf{d}} | d_k, Z_k, W_{[N]}\big)  \label{eq:cacheprivacy1}\\
			&\leq I\big( Q_{1,1},\ldots, Q_{1,K} ; Q_{2,1},\ldots, Q_{2,K}  | d_k, Z_k, W_{[N]}\big) \label{eq:cacheprivacy2}\\
			&=  \sum_{k'\in [K]\setminus \{k\}} I\big( Q_{1,k'} ; Q_{2,k'}  |  W_{[N]}\big) = 0 \label{eq:cache privacy}
		\end{align}
	\end{subequations}
	where again \eqref{eq:cacheprivacy1} comes from~\eqref{eq:independent}, \eqref{eq:cacheprivacy2} comes from the fact that  $\mathscr{M}_{k}=Q_{1,k}, k\in [K]$   and that $ \{Y_{\mathcal{S}}, \mathcal{S} \subseteq [K], |\mathcal{S}|=t+1\}$ is a function of  $\{Q_{2,k}, k \in [K]\}$ and $W_{[N]}$, and \eqref{eq:cache privacy} comes from that $\mathscr{M}_{k}=Q_{1,k}$ is contained in $Z_k$ and that 
	the pairs $(Q_{1,k'}, Q_{2,k'}), k' \in [K]$ are independent of each other and the query independence condition in \eqref{eq:independent queries constriants} holding for the PIR scheme. This completes the proof.	

\section{Two-server PIR Schemes for Theorem~\ref{thm:pirschemes}} \label{sec:schemes}
\label{sub:newPIRschemes}

	For the case $N=2$, we use the proposed PIR scheme in  \cite[Section III-A]{tian2019capacity}. We proceed for other values. Thus part 1 of the theorem is already proved. 
	
	\subsection{$N=3$} \label{sec:schemen3}
	Assume the library has three files $W_1, W_2, W_3$. We define a random variable $T$ which takes value uniformly at random from the set $\{0,1,2\}$. The proposed PIR scheme for different parameter regimes $T$ and demand index $d$ is depicted in Table \ref{tab:tablen3}. 
	\begin{table}
		\centering
		\begin{tabular}{| c || c || c | c | c |} 
			\hline
			\multirow{2}{*}{} & \multirow{2}{*}{Server $1$} &  \multicolumn{3}{c|}{Server $2$} \\ \cline{3-5}
			& & $d=1$ & $d=2$ & $d=3$ \\
			\hline\hline
			$T=0$ & $W_1+W_2$ & $W_2$ & $W_1$ & $W_3$ \\
			\hline
			$T=1$ & $W_1+W_3$ & $W_3$ & $W_2$ & $W_1$ \\
			\hline
			$T=2$ & $W_2+W_3$ & $W_1$ & $W_3$ & $W_2$ \\
			\hline
		\end{tabular}
		\caption{Proposed PIR scheme for $N=3$.}
  \label{tab:tablen3} 
	\end{table}
	
	As one can see in Table \ref{tab:tablen3}, there are 3 different answers for queries sent to server 1 including $W_1+W_2$, $W_1+W_3$, and $W_2+W_3$. We assign query values $Q_1=1$, $Q_1=2$, and $Q_1=3$ to these answers respectively. Similarly we assign query values $Q_2=1$, $Q_2=2$, and $Q_2=3$ for the answers of the second server $W_1$, $W_2$, and $W_3$ respectively. Note  that the queries in the proposed scheme are independent of file realization, so we can remove the terms in the condition from the constraints in \eqref{eq:PIR privacy constraint} and \eqref{eq:independent queries constriants}. 
	
	The query $Q_1 \in \{1,2,3\}$ sent to server 1 is clearly independent of the demand. For the query $Q_2 \in \{1,2,3\}$ sent to server 2 we have
	\begin{align*}
		&\Pr (Q_2=1)  = \Pr (T=0,d=2) + \Pr (T=1,d=3)+\Pr (T=2,d=1)  =1/3.
	\end{align*}  
	In addition, we have
	\begin{align*}
		\Pr (Q_2=1 | d=1) = \Pr (T=2 | d=1) = 1/3.
	\end{align*}
	Hence, we will have $\Pr (Q_2=1)=\Pr (Q_2=1 | d=1)$. Following similarly, we can conclude that $P(Q_2)=P(Q_2|d)$ for all values $Q_2 \in \{1,2,3\}$ and $d \in \{1,2,3\}$ proving \eqref{eq:PIR privacy constraint} to hold. Next we should check the query independence condition in \eqref{eq:independent queries constriants}. We have
	\begin{align*}
		&\Pr(Q_2=1|Q_1=1) = \Pr(Q_2=1|T=0)=1/3=\Pr(Q_2=1).
	\end{align*}
	Again one can similarly show $P(Q_2 |Q_1 )=P(Q_2)$ holds for all $Q_1 \in \{1,2,3\}$ and $Q_2 \in \{1,2,3\}$ proving \eqref{eq:independent queries constriants} to hold. Decodability condition in \eqref{eq:PIR decoding constraint} can be also easily checked to hold. The download cost from each server is $1$ so {\bf the achieved total download cost of this PIR scheme is  $R_D=2$, and the subpacketization is $F^{\prime}=1$.}
	
	For the example of a coded caching with private demands and caches with parameters $N=3, K=2, M=2$, using this PIR scheme in Theorem~\ref{thm:cachepir} with $N=3$ and $t=1$, we get the achieved load of $\frac{1}{2}$ and subpacketization level of $2$. In this example, the achieved load by the virtual users scheme in \cite{DBLP:journals/corr/abs-1909-03324} is $\frac{2}{5}$ and the needed subpacketization level is $15$.  
	
	\subsection{$N=4$} \label{sec:schemen4}
	Assume the library has four files $W_1, W_2, W_3, W_4$. We define a random variable $T$ which takes value uniformly at random from the set $\{0,1,2,3\}$. The proposed PIR scheme for different parameter regimes $T$ and demand index $d$ is depicted in Table \ref{tab:tablen4}. 
	\begin{table*}
		\centering
		\begin{adjustwidth}{-1.7cm}{}
		\begin{tabular}{| c || c || c | c | c | c |} 
			\hline
			\multirow{2}{*}{} & \multirow{2}{*}{Server $1$} &  \multicolumn{4}{c|}{Server $2$} \\ \cline{3-6}
			& & $d=1$ & $d=2$ & $d=3$ & $d=4$ \\
			\hline\hline
			$T=0$ & $W_1+W_2+W_3+W_4$ & $-W_1+W_2+W_3+W_4$ & $W_1-W_2+W_3+W_4$ & $W_1+W_2-W_3+W_4$ & $W_1+W_2+W_3-W_4$ \\
			\hline
			$T=1$ & $-W_1-W_2+W_3+W_4$ & $W_1-W_2+W_3+W_4$ & $-W_1+W_2+W_3+W_4$ & $W_1+W_2+W_3-W_4$ & $W_1+W_2-W_3+W_4$ \\
			\hline
			$T=2$ & $-W_1+W_2-W_3+W_4$ & $W_1+W_2-W_3+W_4$ & $W_1+W_2+W_3-W_4$ & $-W_1+W_2+W_3+W_4$ & $W_1-W_2+W_3+W_4$ \\
			\hline
			$T=3$ & $-W_1+W_2+W_3-W_4$ & $W_1+W_2+W_3-W_4$ & $W_1+W_2-W_3+W_4$ & $W_1-W_2+W_3+W_4$ & $-W_1+W_2+W_3+W_4$ \\
			\hline
		\end{tabular}
		\caption{Proposed PIR scheme for $N=4$.}
  \label{tab:tablen4} 
	\end{adjustwidth}
	\end{table*}

	As one can see in Table \ref{tab:tablen4}, there are 4 different answers for queries sent to server 1 including $W_1+W_2+W_3+W_4$, $-W_1-W_2+W_3+W_4$, $-W_1+W_2-W_3+W_4$, and $-W_1+W_2+W_3-W_4$. We assign query values $Q_1=1$, $Q_1=2$, $Q_1=3$, and $Q_1=4$ to these answers respectively. Similarly we assign query values $Q_2=1$, $Q_2=2$, $Q_2=3$, and $Q_2=4$ for the answers of the second server $-W_1+W_2+W_3+W_4$, $W_1-W_2+W_3+W_4$, $W_1+W_2-W_3+W_4$, and $W_1+W_2+W_3-W_4$ respectively. The queries in the proposed scheme are independent of file realization, so we can remove the terms in the condition from the constraints in \eqref{eq:PIR privacy constraint} and \eqref{eq:independent queries constriants}. 
	
	The query $Q_1 \in \{1,2,3,4\}$ sent to server 1 is clearly independent of the demand. For the query $Q_2 \in \{1,2,3,4\}$ sent to server 2 we have
	\begin{align*}
		&\Pr (Q_2=1)  = \Pr (T=0,d=1) + \Pr (T=1,d=2) +\Pr (T=2,d=3)+\Pr (T=3,d=4) \\
		&=1/4.
	\end{align*}  
	In addition, we have
	\begin{align*}
		\Pr (Q_2=1 | d=1) = \Pr (T=0 | d=1) = 1/4.
	\end{align*}
	Hence, we will have $\Pr (Q_2=1)=\Pr (Q_2=1 | d=1)$. Following similarly, we can conclude that $P(Q_2)=P(Q_2|d)$ for all values $Q_2 \in \{1,2,3,4\}$ and $d \in \{1,2,3,4\}$ proving \eqref{eq:PIR privacy constraint} to hold. Next we should check the query independence condition in \eqref{eq:independent queries constriants}. We have
	\begin{align*}
		&\Pr(Q_2=1|Q_1=1) = \Pr(Q_2=1|T=0)=1/4=\Pr(Q_2=1).
	\end{align*}
	Again one can similarly show $P(Q_2 |Q_1 )=P(Q_2)$ holds for all $Q_1 \in \{1,2,3,4\}$ and $Q_2 \in \{1,2,3,4\}$ proving \eqref{eq:independent queries constriants} to hold. Decodability condition in \eqref{eq:PIR decoding constraint} can be also easily checked to hold. The download cost from each server is $1$ so {\bf the achieved total download cost of this PIR scheme is  $R_D=2$, and the subpacketization is $F^{\prime}=1$.}
	
	For the example of a coded caching with private demands and caches with parameters $N=4, K=2, M=\frac{5}{2}$, using this PIR scheme in Theorem~\ref{thm:cachepir} with $N=4$ and $t=1$, we get the achieved load of $\frac{1}{2}$ and subpacketization level of $2$. In this example, the achieved load by the virtual users scheme in \cite{DBLP:journals/corr/abs-1909-03324} is $\frac{1}{2}$ and the needed subpacketization level is $56$.

	\section{Proof of Theorem~\ref{thm:lowerbound}: Lower Bound on Two-server PIR Schemes Satisfying the UDIQ Condition}
	\label{sub:proof of converse}
	
	\subsection{Proof of Theorem~\ref{thm:lowerbound}-1}
	Without loss of generality, we assume that $\mathcal{Q}_1=\{1,2,\ldots,N_1\}$ and $\mathcal{Q}_2=\{1,2,\ldots,N_2\}$.
	Based on the fact that the queries should not reveal any information about the demand as in \eqref{eq:PIR privacy constraint},  we have
	\begin{align}
		\Pr (d=\tau | Q_1=1) = \ldots = \Pr (d=\tau | Q_1=N_1),  \ \forall \tau \in [N], \label{eq:lower bound step 1}
	\end{align}
	in which $d$ is the demand. Based on the definition in \eqref{def:conditional},  we further extend~\eqref{eq:lower bound step 1} as follows,
	\begin{align} \label{37}
		&\sum_{Q_2 \in \mathcal{U}_{\tau|Q_1=1}} \Pr (Q_2 | Q_1=1) = \ldots = \sum_{Q_2 \in \mathcal{U}_{\tau|Q_1=N_1}} \Pr (Q_2 | Q_1=N_1). 
	\end{align}
	By the independent queries condition in \eqref{eq:independent queries constriants}, and since we have uniform query distribution,
	the values of the probability functions $\Pr (Q_2 | Q_1)$ for all   $Q_2$ and $Q_1$ are the same,   equal to $1/N_2$. So from \eqref{37}  we have
	\begin{align}
		\frac{	\left|\mathcal{U}_{\tau|Q_1=1}\right| }{N_2}  = \ldots = \frac{ \left|\mathcal{U}_{\tau|Q_1=N_1}\right|}{N_2}. 
	\end{align}
	This proves that	
	$\left|\mathcal{U}_{\tau|Q_1=1}\right|  = \cdots = \left|\mathcal{U}_{\tau|Q_1=N_1}\right|$. In addition to $$\Pr (d=1 | Q_1=1)=\Pr (d=2 | Q_1=1)=\cdots= \Pr (d=N| Q_1=1),$$ which comes from the privacy constraint,  we have $\left|\mathcal{U}_{\tau_1|Q_1=1}\right|=\left|\mathcal{U}_{\tau_2|Q_1=1}\right|$ where $\tau_1, \tau_2 \in [N]$.	
	Similarly we also have $\left|\mathcal{U}_{\tau|Q_2=1}\right|  = \cdots = \left|\mathcal{U}_{\tau|Q_2=N_2}\right|$ and 
	$\left|\mathcal{U}_{\tau_1|Q_2=1}\right|  = \left|\mathcal{U}_{\tau_2|Q_2=1}\right|$ where   $\tau_1, \tau_2 \in [N]$.
	 This completes the proof of the first part.
	
	\subsection{Proof of Theorem~\ref{thm:lowerbound}-2}
	Based on the condition of independent queries, for any $q_2 \in \mathcal{Q}_2$, all the queries in $\mathcal{Q}_1$ should be exhausted for all choices of the demanded file index $\tau$. In other words, for any $q_2 \in \mathcal{Q}_2$ we should have
	\begin{align}
		N_1 \leq \left|\mathcal{U}_{\tau=1|Q_2=q_2}\right| + \dots + \left|\mathcal{U}_{\tau=N|Q_2=q_2}\right|=Nn_1,
	\end{align}
	which resluts $\frac{N_1}{n_1} \leq N$. With the same argument we have $\frac{N_2}{n_2} \leq N$.

	\subsection{Proof of Theorem~\ref{thm:lowerbound}-3}
	Based on the definition of $\mathcal{U}_{\tau|Q_1=q_1}$, we have
	\begin{align}
		\mathcal{U}_\tau = \left\{ \left(1, j_1\right) : j_1 \in \mathcal{U}_{\tau | Q_1=1}  \right\} \bigcup \ldots \bigcup \left\{ \left(N_1, j_{N_1}\right): j_{N_1} \in \mathcal{U}_{\tau | Q_1=N_1}  \right\}
	\end{align}
	Since all the sets above are disjoint and of size $n_2$, we have
	\begin{align}
		|\mathcal{U}_\tau| = N_1 n_2,
	\end{align}
	or similarly
	\begin{align}
		|\mathcal{U}_\tau| = N_2 n_1.
	\end{align}
	So we have
	\begin{align}
		|\mathcal{U}_1| + |\mathcal{U}_2| + \ldots + |\mathcal{U}_N| = N N_1 n_2=N N_2 n_1.
	\end{align}
	Since in total we have $N_1 N_2$ different pairs of queries for the two servers, roughly speaking, each query pair should be able to decode $\frac{NN_1n_2}{N_1N_2}=N\frac{n_2}{N_2}=N\frac{n_1}{N_1} \triangleq \alpha N$ files. Thus, we would need at least $\frac{N}{\alpha N} = \frac{1}{\alpha} \triangleq \alpha'$ pairs of queries to cover all the files.  A formal proof will start with the following lemma.
	
	\begin{lemma} \label{lem}
		For $\alpha_1 \in [N_1]$ and $\alpha_2 \in [N_2]$ such that $\alpha_1 \alpha_2 = \left\lceil\frac{N_1}{n_1}\right\rceil=\left\lceil\frac{N_2}{n_2}\right\rceil$, there exist  $\alpha_1$ queries chosen from $\mathcal{Q}_1$ and $\alpha_2$ queries chosen from $\mathcal{Q}_2$ such that the resulting $\alpha_1 \alpha_2$ pairs of queries can recover all the $N$ files.
	\end{lemma}

	\begin{proof}
		We choose $\alpha_1$ queries from $\mathcal{Q}_1$ and $\alpha_2$ queries from $\mathcal{Q}_2$ uniformly at random. Without loss of generality,  we assume that the chosen queries are $\mathcal{Q}_{1,\alpha_1} = \{1,2,\ldots,\alpha_1\}$ and $\mathcal{Q}_{2,\alpha_2} = \{1,2,\ldots,\alpha_2\}$ respectively. We should mention that because of the demand privacy constraint in \eqref{eq:PIR privacy constraint}, all the queries in $\mathcal{Q}_1$ and $\mathcal{Q}_2$ should appear at least once in $\mathcal{U}_\tau$ for any $\tau \in [N]$. For the first query from the first server $Q_1=1$, the probability that $Q_2=1$ is not in the set $\mathcal{U}_{\tau | Q_1=1}$ equals $\Pr (Q_2=1 \notin \mathcal{U}_{\tau | Q_1=1}) = 1-\frac{n_2}{N_2}$. If we know that $Q_2=1 \notin \mathcal{U}_{\tau | Q_1=1}$, the probability that $Q_2=2 \notin \mathcal{U}_{\tau | Q_1=1}$ would be $\Pr (Q_2=2 \notin \mathcal{U}_{\tau | Q_1=1} | Q_2=1 \notin \mathcal{U}_{\tau | Q_1=1}) = 1-\frac{n_2}{N_2-1}$. Similarly continuing, we can compute the probability that none of the $\alpha_2$ queries chosen from $\mathcal{Q}_2$ appears as a pair with $Q_1=1$ in the set $\mathcal{U}_\tau$.
		\begin{align}
			& \Pr \left(\left\{ \left(Q_1=1, \mathcal{Q}_{2,\alpha_2} \right) \right\} \cap \mathcal{U}_{\tau | Q_1=1} = \emptyset\right) \nonumber \\
			& = \left(1-\frac{n_2}{N_2}\right) \left(1-\frac{n_2}{N_2-1}\right) \ldots \left(1-\frac{n_2}{N_2-(\alpha_2-1)}\right) \nonumber \\
			& \leq \left(1-\frac{n_2}{N_2}\right)^{\alpha_2}.
		\end{align}
		Using the same argument for all $\alpha_1$ queries from $\mathcal{Q}_1$, we have
		\begin{subequations}
		\begin{align}
			& \Pr \left( \left(\mathcal{Q}_{1,\alpha_1}, \mathcal{Q}_{2,\alpha_2} \right) \cap \mathcal{U}_\tau = \emptyset \right)   \\
			& \leq \left(\left(1-\frac{n_2}{N_2}\right)^{\alpha_2}\right)^{\alpha_1} = \left(1-\frac{n_2}{N_2}\right)^{\alpha_1 \alpha_2} \leq \left(1-\frac{n_2}{N_2}\right)^{\frac{N_2}{n_2}}   \\
			& \leq 1-\frac{N_2}{n_2}\frac{n_2}{N_2} +  o\left(\frac{n_2}{N_2}\right)  \label{eq:tylor} \\ & = o\left(\frac{n_2}{N_2}\right),
		\end{align}
		\end{subequations}
		where~\eqref{eq:tylor} comes from the Taylor expansion $(1-x)^y = 1-yx+o(x)$. Now we can write the probability that the set $\left(\mathcal{Q}_{1,\alpha_1}, \mathcal{Q}_{2,\alpha_2} \right)$ cannot recover at least one of the $N$ files. 
		\begin{align}
			\Pr \left( \left(\mathcal{Q}_{1,\alpha_1}, \mathcal{Q}_{2,\alpha_2}\right) \cap \mathcal{U}_\tau \neq \emptyset, \forall \tau \in [N]\right) &\geq \left(1-o\left(\frac{n_2}{N_2}\right)\right)^N  > 0
		\end{align}
		
		Since we have chosen our sets of queries randomly and the probability that all the files are covered is greater than zero in a finite probability space, we can conclude that there exists at least one choice of $\alpha_1$ queries from $\mathcal{Q}_1$ and one choice of $\alpha_2$ queries from $\mathcal{Q}_2$ that covers all files.
	\end{proof}

	The proof of the third part of theorem is immediately resulted from Lemma \ref{lem}. As a result of Lemma \ref{lem}, suppose we choose $\alpha_1 \in [N_1]$ queries from $\mathcal{Q}_1$ and $\alpha_2 \in [N_2]$ queries from $\mathcal{Q}_2$ such that the resulting number of pairs $\alpha_1 \alpha_2 = \left\lceil\frac{N_1}{n_1}\right\rceil=\left\lceil\frac{N_2}{n_2}\right\rceil $ can recover all the files. Based on the cut-set bound we have
	\begin{align}
		\alpha_1 R_{D_1} + \alpha_2 R_{D_2} \geq N.
	\end{align}
	Taking the minimum on the left hand-side, proves this part.
	
	\subsection{Proof of Theorem~\ref{thm:lowerbound}-4}
	For the forth part of the theorem, if we assume that $R_{D_1}=R_{D_2}=R'_{D}$ and $\alpha_1 \alpha_2=\lceil\alpha'\rceil$, we will have
	\begin{align}
		R'_{D} \geq \frac{N}{\alpha_1+\alpha_2} = \frac{N}{\alpha_1+\frac{\lceil\alpha'\rceil}{\alpha_1}} = \frac{\alpha_1 N}{\alpha_1^2+\lceil\alpha'\rceil}.
	\end{align}
	Thus we have
	\begin{align} \label{R}
		R'_{D}  \geq \max_{\alpha_1 \in [N_1]} \frac{\alpha_1 N}{\alpha_1^2+\lceil\alpha'\rceil}.
	\end{align}
	To derive the optimum value for $\alpha_1$, we assume that it is continuous and take the derivative of the right hand side with respect to $\alpha_1$ and put it equal to zero. We will have
	\begin{align}
		N(\alpha_1^2+\lceil\alpha'\rceil) = (\alpha_1N)(2\alpha_1),
	\end{align} 
	which leads to $\alpha_1=\sqrt{\lceil\alpha'\rceil}$. Since $\alpha_1$ and $\alpha_2$ should be integers, we can lower bound the right hand-side of \eqref{R} as follows
	\begin{align}
		R'_{D} \geq \frac{N}{2(\sqrt{\lceil\alpha'\rceil}+1)} = \frac{N}{2\left(\sqrt{\left\lceil\frac{N_1}{n_1}\right\rceil}+1\right)} = \frac{N}{2\left(\sqrt{\left\lceil\frac{N_2}{n_2}\right\rceil}+1\right)}.
	\end{align}


\bibliographystyle{IEEEtran}
\bibliography{references}

\begin{thebibliography}{10}
\providecommand{\url}[1]{#1}
\csname url@samestyle\endcsname
\providecommand{\newblock}{\relax}
\providecommand{\bibinfo}[2]{#2}
\providecommand{\BIBentrySTDinterwordspacing}{\spaceskip=0pt\relax}
\providecommand{\BIBentryALTinterwordstretchfactor}{4}
\providecommand{\BIBentryALTinterwordspacing}{\spaceskip=\fontdimen2\font plus
\BIBentryALTinterwordstretchfactor\fontdimen3\font minus
  \fontdimen4\font\relax}
\providecommand{\BIBforeignlanguage}[2]{{%
\expandafter\ifx\csname l@#1\endcsname\relax
\typeout{** WARNING: IEEEtran.bst: No hyphenation pattern has been}%
\typeout{** loaded for the language `#1'. Using the pattern for}%
\typeout{** the default language instead.}%
\else
\language=\csname l@#1\endcsname
\fi
#2}}
\providecommand{\BIBdecl}{\relax}
\BIBdecl

\bibitem{gholami2022coded}
A.~Gholami, K.~Wan, H.~Sun, M.~Ji, and G.~Caire, ``Coded caching with private
  demands and caches,'' in \emph{2022 IEEE International Symposium on
  Information Theory (ISIT)}.\hskip 1em plus 0.5em minus 0.4em\relax IEEE,
  2022, pp. 1396--1401.

\bibitem{maddah2014fundamental}
M.~A. Maddah-Ali and U.~Niesen, ``Fundamental limits of caching,'' \emph{IEEE
  Transactions on information theory}, vol.~60, no.~5, pp. 2856--2867, 2014.

\bibitem{yu2017exact}
Q.~Yu, M.~A. Maddah-Ali, and A.~S. Avestimehr, ``The exact rate-memory tradeoff
  for caching with uncoded prefetching,'' \emph{IEEE Transactions on
  Information Theory}, vol.~64, no.~2, pp. 1281--1296, 2017.

\bibitem{wan2016optimality}
K.~Wan, D.~Tuninetti, and P.~Piantanida, ``On the optimality of uncoded cache
  placement,'' in \emph{2016 IEEE Information Theory Workshop (ITW)}.\hskip 1em
  plus 0.5em minus 0.4em\relax IEEE, 2016, pp. 161--165.

\bibitem{maddah2014decentralized}
M.~A. Maddah-Ali and U.~Niesen, ``Decentralized coded caching attains
  order-optimal memory-rate tradeoff,'' \emph{IEEE/ACM Transactions On
  Networking}, vol.~23, no.~4, pp. 1029--1040, 2014.

\bibitem{pedarsani2015online}
R.~Pedarsani, M.~A. Maddah-Ali, and U.~Niesen, ``Online coded caching,''
  \emph{IEEE/ACM Transactions on Networking}, vol.~24, no.~2, pp. 836--845,
  2015.

\bibitem{ji2015fundamental}
M.~Ji, G.~Caire, and A.~F. Molisch, ``Fundamental limits of caching in wireless
  d2d networks,'' \emph{IEEE Transactions on Information Theory}, vol.~62,
  no.~2, pp. 849--869, 2015.

\bibitem{ji2017order}
M.~Ji, A.~M. Tulino, J.~Llorca, and G.~Caire, ``Order-optimal rate of caching
  and coded multicasting with random demands,'' \emph{IEEE Transactions on
  Information Theory}, vol.~63, no.~6, pp. 3923--3949, 2017.

\bibitem{niesen2016coded}
U.~Niesen and M.~A. Maddah-Ali, ``Coded caching with nonuniform demands,''
  \emph{IEEE Transactions on Information Theory}, vol.~63, no.~2, pp.
  1146--1158, 2016.

\bibitem{karamchandani2016hierarchical}
N.~Karamchandani, U.~Niesen, M.~A. Maddah-Ali, and S.~N. Diggavi,
  ``Hierarchical coded caching,'' \emph{IEEE Transactions on Information
  Theory}, vol.~62, no.~6, pp. 3212--3229, 2016.

\bibitem{shanmugam2016finite}
K.~Shanmugam, M.~Ji, A.~M. Tulino, J.~Llorca, and A.~G. Dimakis,
  ``Finite-length analysis of caching-aided coded multicasting,'' \emph{IEEE
  Transactions on Information Theory}, vol.~62, no.~10, pp. 5524--5537, 2016.

\bibitem{jin2019new}
S.~Jin, Y.~Cui, H.~Liu, and G.~Caire, ``A new order-optimal decentralized coded
  caching scheme with good performance in the finite file size regime,''
  \emph{IEEE Transactions on Communications}, vol.~67, no.~8, pp. 5297--5310,
  2019.

\bibitem{yan2017placement}
Q.~Yan, M.~Cheng, X.~Tang, and Q.~Chen, ``On the placement delivery array
  design for centralized coded caching scheme,'' \emph{IEEE Transactions on
  Information Theory}, vol.~63, no.~9, pp. 5821--5833, 2017.

\bibitem{wang2019placement}
J.~Wang, M.~Cheng, Q.~Yan, and X.~Tang, ``Placement delivery array design for
  coded caching scheme in d2d networks,'' \emph{IEEE Transactions on
  Communications}, vol.~67, no.~5, pp. 3388--3395, 2019.

\bibitem{yan2018placement}
Q.~Yan, M.~Wigger, and S.~Yang, ``Placement delivery array design for
  combination networks with edge caching,'' in \emph{2018 IEEE International
  Symposium on Information Theory (ISIT)}.\hskip 1em plus 0.5em minus
  0.4em\relax IEEE, 2018, pp. 1555--1559.

\bibitem{sasi2021multi}
S.~Sasi and B.~S. Rajan, ``Multi-access coded caching scheme with linear
  sub-packetization using pdas,'' \emph{IEEE Transactions on Communications},
  vol.~69, no.~12, pp. 7974--7985, 2021.

\bibitem{cheng2021framework}
M.~Cheng, J.~Wang, X.~Zhong, and Q.~Wang, ``A framework of constructing
  placement delivery arrays for centralized coded caching,'' \emph{IEEE
  Transactions on Information Theory}, vol.~67, no.~11, pp. 7121--7131, 2021.

\bibitem{zhong2020placement}
X.~Zhong, M.~Cheng, and J.~Jiang, ``Placement delivery array based on
  concatenating construction,'' \emph{IEEE Communications Letters}, vol.~24,
  no.~6, pp. 1216--1220, 2020.

\bibitem{shangguan2018centralized}
C.~Shangguan, Y.~Zhang, and G.~Ge, ``Centralized coded caching schemes: A
  hypergraph theoretical approach,'' \emph{IEEE Transactions on Information
  Theory}, vol.~64, no.~8, pp. 5755--5766, 2018.

\bibitem{shanmugam2017coded}
K.~Shanmugam, A.~M. Tulino, and A.~G. Dimakis, ``Coded caching with linear
  subpacketization is possible using ruzsa-szem{\'e}redi graphs,'' in
  \emph{2017 IEEE International Symposium on Information Theory (ISIT)}.\hskip
  1em plus 0.5em minus 0.4em\relax IEEE, 2017, pp. 1237--1241.

\bibitem{yan2017placementbipartite}
Q.~Yan, X.~Tang, Q.~Chen, and M.~Cheng, ``Placement delivery array design
  through strong edge coloring of bipartite graphs,'' \emph{IEEE Communications
  Letters}, vol.~22, no.~2, pp. 236--239, 2017.

\bibitem{tang2018coded}
L.~Tang and A.~Ramamoorthy, ``Coded caching schemes with reduced
  subpacketization from linear block codes,'' \emph{IEEE Transactions on
  Information Theory}, vol.~64, no.~4, pp. 3099--3120, 2018.

\bibitem{somevariantCheng}
M.~Cheng, J.~Jiang, X.~Tang, and Q.~Yan, ``Some variant of known coded caching
  schemes with good performance,'' \emph{IEEE Transactions on Communications},
  vol.~68, no.~3, pp. 1370--1377, 2020.

\bibitem{wan2020coded}
K.~Wan and G.~Caire, ``On coded caching with private demands,'' \emph{IEEE
  Transactions on Information Theory}, vol.~67, no.~1, pp. 358--372, 2020.

\bibitem{engelmann2017content}
F.~Engelmann and P.~Elia, ``A content-delivery protocol, exploiting the privacy
  benefits of coded caching,'' in \emph{2017 15th International Symposium on
  Modeling and Optimization in Mobile, Ad Hoc, and Wireless Networks
  (WiOpt)}.\hskip 1em plus 0.5em minus 0.4em\relax IEEE, 2017, pp. 1--6.

\bibitem{DBLP:journals/corr/abs-1909-03324}
\BIBentryALTinterwordspacing
S.~Kamath, ``Demand private coded caching,'' \emph{CoRR}, vol. abs/1909.03324,
  2019. [Online]. Available: \url{http://arxiv.org/abs/1909.03324}
\BIBentrySTDinterwordspacing

\bibitem{aravind2020coded}
V.~Aravind, P.~K. Sarvepalli, and A.~Thangaraj, ``Coded caching with demand
  privacy: Constructions for lower subpacketization and generalizations,''
  \emph{arXiv preprint arXiv:2007.07475}, 2020.

\bibitem{yan2021fundamental}
Q.~Yan and D.~Tuninetti, ``Fundamental limits of caching for demand privacy
  against colluding users,'' \emph{IEEE Journal on Selected Areas in
  Information Theory}, vol.~2, no.~1, pp. 192--207, 2021.

\bibitem{wan2021optimal}
K.~Wan, H.~Sun, M.~Ji, D.~Tuninetti, and G.~Caire, ``On the optimal load-memory
  tradeoff of cache-aided scalar linear function retrieval,'' \emph{IEEE
  Transactions on Information Theory}, vol.~67, no.~6, pp. 4001--4018, 2021.

\bibitem{kamath2020demand}
S.~Kamath, J.~Ravi, and B.~K. Dey, ``Demand-private coded caching and the exact
  trade-off for n= k= 2,'' in \emph{2020 National Conference on Communications
  (NCC)}.\hskip 1em plus 0.5em minus 0.4em\relax IEEE, 2020, pp. 1--6.

\bibitem{gurjarpadhye2022fundamental}
C.~Gurjarpadhye, J.~Ravi, S.~Kamath, B.~K. Dey, and N.~Karamchandani,
  ``Fundamental limits of demand-private coded caching,'' \emph{IEEE
  Transactions on Information Theory}, 2022.

\bibitem{aravind2020subpacketization}
V.~Aravind, P.~K. Sarvepalli, and A.~Thangaraj, ``Subpacketization in coded
  caching with demand privacy,'' in \emph{2020 National Conference on
  Communications (NCC)}.\hskip 1em plus 0.5em minus 0.4em\relax IEEE, 2020, pp.
  1--6.

\bibitem{chor1995private}
B.~Chor, O.~Goldreich, E.~Kushilevitz, and M.~Sudan, ``Private information
  retrieval,'' in \emph{Proceedings of IEEE 36th Annual Foundations of Computer
  Science}.\hskip 1em plus 0.5em minus 0.4em\relax IEEE, 1995, pp. 41--50.

\bibitem{cohen1997covering}
G.~Cohen, I.~Honkala, S.~Litsyn, and A.~Lobstein, \emph{Covering codes}.\hskip
  1em plus 0.5em minus 0.4em\relax Elsevier, 1997.

\bibitem{ambainis1997upper}
A.~Ambainis, ``Upper bound on the communication complexity of private
  information retrieval,'' in \emph{International Colloquium on Automata,
  Languages, and Programming}.\hskip 1em plus 0.5em minus 0.4em\relax Springer,
  1997, pp. 401--407.

\bibitem{itoh1999efficient}
T.~Itoh, ``Efficient private information retrieval,'' \emph{IEICE TRANSACTIONS
  on Fundamentals of Electronics, Communications and Computer Sciences},
  vol.~82, no.~1, pp. 11--20, 1999.

\bibitem{beimel2001information}
A.~Beimel and Y.~Ishai, ``Information-theoretic private information retrieval:
  A unified construction,'' in \emph{International Colloquium on Automata,
  Languages, and Programming}.\hskip 1em plus 0.5em minus 0.4em\relax Springer,
  2001, pp. 912--926.

\bibitem{razborov2006omega}
A.~A. Razborov and S.~Yekhanin, ``An$\backslash$omega (n\^{} 1/3) lower bound
  for bilinear group based private information retrieval,'' in \emph{2006 47th
  Annual IEEE Symposium on Foundations of Computer Science (FOCS'06)}.\hskip
  1em plus 0.5em minus 0.4em\relax IEEE, 2006, pp. 739--748.

\bibitem{wehner2005improved}
S.~Wehner and R.~d. Wolf, ``Improved lower bounds for locally decodable codes
  and private information retrieval,'' in \emph{International Colloquium on
  Automata, Languages, and Programming}.\hskip 1em plus 0.5em minus 0.4em\relax
  Springer, 2005, pp. 1424--1436.

\bibitem{chakrabarti2007nearly}
A.~Chakrabarti and A.~Shubina, ``Nearly private information retrieval,'' in
  \emph{International Symposium on Mathematical Foundations of Computer
  Science}.\hskip 1em plus 0.5em minus 0.4em\relax Springer, 2007, pp.
  383--393.

\bibitem{beigel2003nearly}
R.~Beigel, L.~Fortnow, and W.~Gasarch, ``A nearly tight lower bound for private
  information retrieval protocols,'' \emph{Electronic Colloquim on
  Computational Complexity (ECCC)}, 2003.

\bibitem{woodruff2005geometric}
D.~Woodruff and S.~Yekhanin, ``A geometric approach to information-theoretic
  private information retrieval,'' in \emph{20th Annual IEEE Conference on
  Computational Complexity (CCC'05)}.\hskip 1em plus 0.5em minus 0.4em\relax
  IEEE, 2005, pp. 275--284.

\bibitem{dvir20162}
Z.~Dvir and S.~Gopi, ``2-server pir with subpolynomial communication,''
  \emph{Journal of the ACM (JACM)}, vol.~63, no.~4, pp. 1--15, 2016.

\bibitem{efremenko20123}
K.~Efremenko, ``3-query locally decodable codes of subexponential length,''
  \emph{SIAM Journal on Computing}, vol.~41, no.~6, pp. 1694--1703, 2012.

\bibitem{yekhanin2008towards}
S.~Yekhanin, ``Towards 3-query locally decodable codes of subexponential
  length,'' \emph{Journal of the ACM (JACM)}, vol.~55, no.~1, pp. 1--16, 2008.

\bibitem{sun2017capacity}
H.~Sun and S.~A. Jafar, ``The capacity of private information retrieval,''
  \emph{IEEE Transactions on Information Theory}, vol.~63, no.~7, pp.
  4075--4088, 2017.

\bibitem{tian2019capacity}
C.~Tian, H.~Sun, and J.~Chen, ``Capacity-achieving private information
  retrieval codes with optimal message size and upload cost,'' \emph{IEEE
  Transactions on Information Theory}, vol.~65, no.~11, pp. 7613--7627, 2019.

\bibitem{banawan2018multi}
K.~Banawan and S.~Ulukus, ``Multi-message private information retrieval:
  Capacity results and near-optimal schemes,'' \emph{IEEE Transactions on
  Information Theory}, vol.~64, no.~10, pp. 6842--6862, 2018.

\bibitem{sun2018capacity}
H.~Sun and S.~A. Jafar, ``The capacity of symmetric private information
  retrieval,'' \emph{IEEE Transactions on Information Theory}, vol.~65, no.~1,
  pp. 322--329, 2018.

\bibitem{kadhe2019private}
S.~Kadhe, B.~Garcia, A.~Heidarzadeh, S.~El~Rouayheb, and A.~Sprintson,
  ``Private information retrieval with side information,'' \emph{IEEE
  Transactions on Information Theory}, vol.~66, no.~4, pp. 2032--2043, 2019.

\bibitem{tajeddine2018private}
R.~Tajeddine, O.~W. Gnilke, and S.~El~Rouayheb, ``Private information retrieval
  from mds coded data in distributed storage systems,'' \emph{IEEE Transactions
  on Information Theory}, vol.~64, no.~11, pp. 7081--7093, 2018.

\bibitem{tandon2017capacity}
R.~Tandon, ``The capacity of cache aided private information retrieval,'' in
  \emph{2017 55th Annual Allerton Conference on Communication, Control, and
  Computing (Allerton)}.\hskip 1em plus 0.5em minus 0.4em\relax IEEE, 2017, pp.
  1078--1082.

\bibitem{shariatpanahi2018multi}
S.~P. Shariatpanahi, M.~J. Siavoshani, and M.~A. Maddah-Ali, ``Multi-message
  private information retrieval with private side information,'' in \emph{2018
  IEEE Information Theory Workshop (ITW)}.\hskip 1em plus 0.5em minus
  0.4em\relax IEEE, 2018, pp. 1--5.

\bibitem{sun2017collcapacity}
H.~Sun and S.~A. Jafar, ``The capacity of robust private information retrieval
  with colluding databases,'' \emph{IEEE Transactions on Information Theory},
  vol.~64, no.~4, pp. 2361--2370, 2017.

\bibitem{banawan2018capacity}
K.~Banawan and S.~Ulukus, ``The capacity of private information retrieval from
  coded databases,'' \emph{IEEE Transactions on Information Theory}, vol.~64,
  no.~3, pp. 1945--1956, 2018.

\bibitem{shah2014one}
N.~B. Shah, K.~Rashmi, and K.~Ramchandran, ``One extra bit of download ensures
  perfectly private information retrieval,'' in \emph{2014 IEEE International
  Symposium on Information Theory}.\hskip 1em plus 0.5em minus 0.4em\relax
  IEEE, 2014, pp. 856--860.

\bibitem{samy2019capacity}
I.~Samy, R.~Tandon, and L.~Lazos, ``On the capacity of leaky private
  information retrieval,'' in \emph{2019 IEEE International Symposium on
  Information Theory (ISIT)}.\hskip 1em plus 0.5em minus 0.4em\relax IEEE,
  2019, pp. 1262--1266.

\bibitem{dwork2008differential}
C.~Dwork, ``Differential privacy: A survey of results,'' in \emph{International
  conference on theory and applications of models of computation}.\hskip 1em
  plus 0.5em minus 0.4em\relax Springer, 2008, pp. 1--19.

\bibitem{lin2019weakly}
H.-Y. Lin, S.~Kumar, E.~Rosnes, A.~G. i~Amat, and E.~Yaakobi, ``Weakly-private
  information retrieval,'' in \emph{2019 IEEE International Symposium on
  Information Theory (ISIT)}.\hskip 1em plus 0.5em minus 0.4em\relax IEEE,
  2019, pp. 1257--1261.

\bibitem{lin2021capacity}
------, ``The capacity of single-server weakly-private information retrieval,''
  \emph{IEEE Journal on Selected Areas in Information Theory}, vol.~2, no.~1,
  pp. 415--427, 2021.

\bibitem{samy2021asymmetric}
I.~Samy, M.~Attia, R.~Tandon, and L.~Lazos, ``Asymmetric leaky private
  information retrieval,'' \emph{IEEE Transactions on Information Theory},
  vol.~67, no.~8, pp. 5352--5369, 2021.

\bibitem{guo2020information}
T.~Guo, R.~Zhou, and C.~Tian, ``On the information leakage in private
  information retrieval systems,'' \emph{IEEE Transactions on Information
  Forensics and Security}, vol.~15, pp. 2999--3012, 2020.

\bibitem{zhou2020weakly}
R.~Zhou, T.~Guo, and C.~Tian, ``Weakly private information retrieval under the
  maximal leakage metric,'' in \emph{2020 IEEE International Symposium on
  Information Theory (ISIT)}.\hskip 1em plus 0.5em minus 0.4em\relax IEEE,
  2020, pp. 1089--1094.

\bibitem{el2011network}
A.~El~Gamal and Y.-H. Kim, \emph{Network information theory}.\hskip 1em plus
  0.5em minus 0.4em\relax Cambridge university press, 2011.

\bibitem{grolmusz2000superpolynomial}
V.~Grolmusz, ``Superpolynomial size set-systems with restricted intersections
  mod 6 and explicit ramsey graphs,'' \emph{Combinatorica}, vol.~20, no.~1, pp.
  71--86, 2000.

\end{thebibliography}

\end{document}